\documentclass[11pt,draftcls,onecolumn]{IEEEtran} 

\usepackage{cite}
\usepackage{moreverb}
\usepackage{epsfig}
\usepackage{booktabs}
\usepackage{array}
\usepackage{multirow}
\usepackage{threeparttable}
\usepackage{color}
\usepackage{url}
\usepackage{dsfont}
\usepackage{amsmath,bm}
\usepackage{amsfonts}
\usepackage{amssymb} 
\usepackage{makeidx} 
\usepackage{graphicx} 

\newtheorem{proposition}{Proposition}

\DeclareMathAlphabet{\mathpzc}{OT1}{pzc}{m}{it}

\def\bg{{\mathbf g}}

\def\bu{{\mathbf u}}

\def\bx{{\mathbf x}}

\def\bH{{\mathbf H}}
\def\bI{{\mathbf I}}

\def\b0{{\mathbf 0}}

\title {Multiuser MIMO Downlink Beamforming Design Based on Group Maximum SINR Filtering}

\author{Yu-Han Yang,
Shih-Chun Lin \IEEEmembership{Member, IEEE}, and Hsuan-Jung Su
\IEEEmembership{Member, IEEE}
\thanks{The material in this paper was presented in part at IEEE International Conference on Communications (ICC) 2009, Dresden, Germany, and ICC 2008, Beijing, China.}
\thanks{Yu-Han Yang was with the Graduate Institute of Communication Engineering, National Taiwan University. He is now with the Department of Electrical and Computer Engineering, University of Maryland, College Park, USA (e-mail: yhyang@umd.edu).}
\thanks{Shih-Chun Lin is with the Institute of Communications Engineering, National Tsing Hua University, Hsinchu, Taiwan, 30013 (e-mail: linsc@mx.nthu.edu.tw).}
\thanks{Hsuan-Jung Su is with the Department of Electrical Engineering and Graduate Institute of Communication Engineering, National Taiwan University, Taipei, Taiwan, 10617 (e-mail: hjsu@cc.ee.ntu.edu.tw).}
}

\begin{document}
\date{}
\maketitle
\IEEEpeerreviewmaketitle

\begin{abstract}
In this paper we aim to solve the multiuser multi-input
multi-output (MIMO) downlink beamforming problem where one
multi-antenna base station broadcasts data to many users. Each
user is assigned multiple data streams and has multiple antennas
at its receiver. Efficient solutions to the joint transmit-receive
beamforming and power allocation problem based on iterative
methods are proposed. We adopt the group maximum
signal-to-interference-plus-noise-ratio (SINR) filter bank
(GSINR-FB) as our beamformer which exploits receiver diversity
through cooperation between the data streams of a user. The data
streams for each user are subject to an average SINR constraint,
which has many important applications in wireless communication
systems and serves as a good metric to measure the quality of
service (QoS). The GSINR-FB also optimizes the average SINR of its
output. Based on the GSINR-FB beamformer, we find an SINR
balancing structure for optimal power allocation which simplifies
the complicated power allocation problem to a linear one.
Simulation results verify the superiority of the proposed
algorithms over previous works with approximately the same
complexity.
\end{abstract}

\begin{center}
   {\underline{\bf \small EDICS}} \hspace{3mm} {\small MSP-MULT}
\end{center}

\IEEEpeerreviewmaketitle


\section{Introduction}
In this paper, the joint beamforming and power allocation
optimization problem for the multiuser multi-input multi-output
(MIMO) downlink channel is considered. In this system, transmit
and receive beamformings are used to suppress the multiuser
interference and exploit the multi-antenna diversity. Power
allocation at the transmitter is performed to efficiently utilize
the available transmission power. Such a joint beamforming and
power allocation problem has been studied by many researchers
\cite{rashid1998transmit,Bourdoux02,Choi04,Sch04,kha06}. In
\cite{Bourdoux02}\cite{Choi04}, block diagonalization (BD) was
proposed to block-diagonalize the overall channel so that the
multiuser interference at each receiver is thoroughly eliminated.
Such a zero-forcing approach suffers from the noise enhancement
problem, because it removes the multiuser interference by ignoring
the noise. Hence the performance can be improved if the balance
between multiuser interference suppression and noise enhancement
can be found \cite{Sch04}\cite{kha06}.

Under individual signal-to-interference-plus-noise-ratio (SINR)
constraints for users, Schubert and Boche studied the situation
where each user has only one data stream and \emph{single receive
antenna} \cite{Sch04}. It was shown that the optimal solution can
be efficiently found by iterative algorithms. Khachan \emph{et
al.} \cite{kha06} generalized the scheme in \cite{Sch04} to allow
several transmission beams to be grouped to serve a user, and each
user has multiple receiver antennas \cite{kha06}. However, each
data stream is processed separately. Thus, in addition to the
multiuser interference from the other users, there is intra-group
interference between the data streams of a user. This drawback
motivates our work to use a more sophisticated receiver processing
to tackle the intra-group interference.

In this work, we adopt the group maximum SINR filter bank
(GSINR-FB) proposed by \cite{HJSu02} as the beamformer, which
collects the desired signal energy in the streams of each user and
maximize the total SINR at its output. That is, the GSINR-FB lets
these streams cooperate while the filters in \cite{kha06} let them
compete. Based on the GSINR-FB beamformer, we consider a system
which uses the average SINR over data streams for a user as a
metric to measure the quality-of-service (QoS). This criterion is
very useful in many communication scenarios
\cite{JHWang08,Davision2006robust,HJSu02} including the celebrated
space-time block coded systems. It will be shown that the GSINR-FB
based beamformer does improve the performance over the scheme in
\cite{kha06}. Moreover, we find that the SINR balancing structure
exists for this beamforming method, that is, the optimal power
allocation results in the same SINR to target ratio for all users
with the GSINR-FB based beamforming. As will be shown later, this
property makes solving the complicated power allocation problem
much easier. Our work can be seen as a non-trivial generalization
of \cite{Sch04} to the multi-antenna setting which also subsumes
\cite{kha06} as a special case (with independent processing of
data streams). For simplicity, we will first consider group power
allocation which restricts equal power on the data streams of each
user to benefits from the low-complexity power allocation schemes
similar to those in \cite{Sch04}\cite{kha06}. This restriction is
later relaxed by allowing the power of individual data streams to
be adjustable. Besides the GSINR-FB based beamforming, this per
stream power allocation scheme is new compared with
\cite{Sch04}\cite{kha06} and has better performance than the group
power allocation. These two techniques are the key ingredients to
make our performance better than that in \cite{kha06}. With
approximately the same complexity as \cite{kha06}, our approach
exhibits a better performance compared to the existing methods in
\cite{kha06} and the BD based methods.

We will investigate two optimization problems. One is minimizing
the total transmitted power while satisfying a set of average SINR
targets. The other is maximizing the achieved average SINR to
target ratio under a total power constraint.  Based on the
uplink-downlink duality \cite{Tse02}, our methods iteratively
calculate the GSINR-FB based beamforming and power allocation
matrices. The rest of the paper is organized as follows. The
system model and problem formulation are introduced in Section
\ref{chap:ProblemFormulation}. We also briefly discuss the basic
design concept of our iterative algorithms in this section.
Backgrounds such as the GSINR-FB based beamformers and the
applications of the average SINR criterion are provided in Section
\ref{sec:GSINR-FB}. Section \ref{sec:PowerAlloc} presents our
power allocation results. The numerical results are given in
Section \ref{sec:SimulationResult}, and the computational
complexity issues are discussed in Section \ref{sec:Discussions}.
Finally, we give the conclusion in \ref{sec:Conclusion}.

\section{Problem
Formulation and Efficient Iterative
Solutions}\label{chap:ProblemFormulation}
\subsection{Notations}
In this paper, vectors and matrices are denoted in bold-face lower
and upper cases, respectively. For vector $\bg$, $\bg \ge_e 0$
means that every element of $\bg$ is nonnegative. For matrix
$\mathbf{G}$, $\mathrm{trace}(\mathbf{G)}$ denotes the trace;
$\mathbf{G}^{\mathrm{T}}$ and $\mathbf{G}^{H}$ denote the
transpose and Hermitian operations, respectively. $\left\| \cdot
\right\|_F$ denotes the Frobenius norm, which is defined as
$\left\| {\bf{G}} \right\|_F  = \sqrt {{\rm{trace}}\left(
{{\bf{GG}}^H } \right)}$. $\mathbf{G}_s^{-1}$ and $|\mathbf{G}_s|$
are, respectively, the inverse and determinant of a square matrix
$\mathbf{G}_s.$ And $\mathbf{I}_{n}$ denotes the identity matrix
of dimension $n$. A diagonal matrix is denoted $diag \{ \ldots \}$
whose $k$th parameter is the $k$th diagonal term in the matrix.
$E[\cdot]$ denotes the expectation operator.

\subsection{System Model}
\label{sec:SystemModel} Consider the downlink scenario with $K$
users, where a base station is equipped with $M$ antennas. The
upper part of Fig. \ref{fig:DL_UL} shows the overall system block
diagram for user $k$, who has $N_k$ receive antennas and receives
$L_k$ data streams, where $L_k$ satisfies the constraint $L_k \le
\min\left\{M,N_k\right\}$ to make sure effective recovery of the
data streams at the receiver. Thus the $K$ users have a total of
$N=\sum\nolimits_{k = 1}^K {N_k }$ receive antennas receiving a
total of $L=\sum\nolimits_{k = 1}^K {L_k }$ grouped data streams.
For a given symbol time, the data streams intended for user $k$
are denoted by a vector of symbols $\mathbf{x}_k=[x_{k1}, x_{k2},
..., x_{kL_k}]^T$. The $L$ data streams are concatenated in a
vector $\mathbf{x} = [\mathbf{x}^T_1,...,\mathbf{x}^T_K]^T$.
Without loss of generality, we assume that $\bx$ is zero mean with
covariance matrix $\bI_L$. The precoder $\mathbf{U}_k\in
\mathcal{C}^{M\times L_k}$ processes user $k$'s data streams
before they are transmitted over the $M$ antennas. These
individual precoders together form the $M\times L$ global
transmitter beamforming matrix
$\mathbf{U}=[\mathbf{U}_1,\mathbf{U}_2,...,\mathbf{U}_K]$. The
power allocation matrix for user $k$ is a diagonal matrix
\begin{equation}
\mathbf{P}_k = diag\{p_{k1}, p_{k2},..., p_{kL_k}\},
\end{equation}
where $p_{kj}$ is the power allocated to the $j$th data stream of user $k$, and the global power allocation matrix
\begin{equation}
\mathbf{P}=diag\{\mathbf{P}_1, \mathbf{P}_2,..., \mathbf{P}_K\}
\end{equation}
is a block diagonal matrix of dimension $L\times L$. The
transmitter broadcasts signals $\mathbf{U\sqrt{P}x}$ to all of the
$K$ users.

User $k$ receives a length $N_k$ vector
${\bf{y}}_k=\mathbf{H}^H_k\mathbf{U\sqrt{P}x}$, which can be
expanded as
\begin{equation}
 \label{equ:DL_ReceiveVector}
 {\bf{y}}_k = {\bf{H}}_k^H {\bf{U}}_k \sqrt {{\bf{P}}_k } {\bf{x}}_k  +
{\bf{H}}_k^H
 \left( {\sum\limits_{j \ne k, j=1}^K {{\bf{U}}_j \sqrt {{\bf{P}}_j } {\bf{x}}_j } } \right) + {\bf{n}}_k,
\end{equation}
where the channel between the transmitter and user $k$ is
represented by the $N_k\times M$ matrix $\mathbf{H}_k^H$, the
Hermitian of $\mathbf{H}_k$; $\mathbf{n}_k$ represents the
zero-mean additive white Gaussian noise (AWGN) at user $k$'s
receive antennas with variance $\sigma^2$ per antenna and the
covariance matrix
$E[\mathbf{n}_k\mathbf{n}_k^H]=\sigma^2\mathbf{I}_{N_k}$. The
resulting $N\times M$ global channel matrix is $\mathbf{H}^H$,
with $\mathbf{H}=[\mathbf{H}_1,\mathbf{H}_2,...,\mathbf{H}_K]$. We
assume that the transmitter has perfect knowledge of the channel
matrix $\mathbf{H}$, and receiver $k$ knows its $\mathbf{H}_k$
perfectly. The second term on the right-hand-side of
(\ref{equ:DL_ReceiveVector}) is the inter-group multiple user
interference for user $k$. To estimate its $L_k$ symbols
$\mathbf{x}_k$, user $k$ processes $\mathbf{y}_k$ with its
$L_k\times N_k$ receive beamforming matrix $\mathbf{V}_k^H$. The
resulting estimated signal vector is
\begin{align} \notag
 {\mathbf{\hat
x}}_k &={\mathbf{V}}_k^H {\mathbf{H}}_k^H {\mathbf{U\sqrt{P}x}} +
{\mathbf{V}}_k^H {\mathbf{n}}_k \\
 & =  \mathbf{V}_k^H {\bf{H}}_k^H {\bf{U}}_k \sqrt {{\bf{P}}_k } {\bf{x}}_k  + \mathbf{V}_k^H {\bf{H}}_k^H
     \left( {\sum\limits_{j \ne k, j=1}^K {{\bf{U}}_j \sqrt {{\bf{P}}_j } {\bf{x}}_j } }
     \right)+\mathbf{V}_k^H {\bf{n}}_k. \label{equ:DL_AfterV}
\end{align}
Without loss of generality, as \cite{HJSu02}, we assume that the
interference-plus-noise components of the filter bank output in
\eqref{equ:DL_AfterV} are uncorrelated. For any filter bank that
produces correlated components, one can easily find another filter
bank which makes these component uncorrelated but with the same
performance. The details can be found in \cite{HJSu02}.

Finally, owing to the non-cooperative nature between users in
broadcast channels, the global receiver beamforming filter
$\mathbf{V}^H$, formed by collecting the individual receiver
filters, is a block diagonal matrix of dimension $L\times N$ where
$\mathbf{V}=diag\{[\mathbf{V}_1,\mathbf{V}_2,...,\mathbf{V}_K]\}$.

\subsection{Problem Formulation} \label{sec:problem}
In this paper, we consider the average SINR of user $k$ over all
its $L_k$ data streams ${\overline{ \rm SINR}}_k =
\sum_{j=1}^{L_k}{\rm SINR}_{kj}/L_k$ as the performance measure,
where ${\rm SINR}_{kj}$ is the $\rm SINR$ of the $j$th data stream
of user $k$. The importance and applications of this design
criterion will be reviewed in detail later in Section
\ref{secASINRapp}. Based on the average SINR constraints and
system model described in Section \ref{sec:SystemModel}, we
consider two problems as follows.  The first optimization problem,
which will be referred to as Problem Pr in the following sections
is \\ \\ \noindent \textbf{Problem Pr}: Given a total power
constraint $P_{\rm max}$ and the SINR target $\gamma _k$ for user
$k$, maximize $\mathop {\min }\limits_k {\overline{\rm{SINR}}}_k
/\gamma _k$ over all beamformers $\mathbf{U}$, $\mathbf{V}$, and
power allocation matrix $\mathbf{P}$, that is,
\begin{equation}
\label{equ:P1_sumPowerConstraint} \mathop {\max
}\limits_{{\bf{U}},{\bf{V}},{\bf{P}}} \mathop {\min }\limits_k
\frac{{{\overline{\rm{SINR}}}_k }}{{\gamma _k }} {\rm{\ \ subj.\
to\ }} \sum\limits_{k = 1}^K {\sum\limits_{j = 1}^{L_k } {p_{kj} }
}  \leq P_{\max }.
\end{equation}
We call ${\rm{\overline{SINR}}}_k/\gamma_k$ the SINR to target
ratio for user $k$.\\

If the minimum SINR to target ratio in Equation
(\ref{equ:P1_sumPowerConstraint}) can be made greater than or
equal to one, then the second optimization problem is to find the
minimum power required such that the SINR targets can be all
satisfied. The mathematical formulation of this problem, which
will be referred to as Problem Pp in the following sections is
\\ \\ \noindent
\textbf{Problem Pp}: Given a constraint on the minimum SINR to target ratio, minimize
the total transmitted power over all beamformers $\mathbf{U}$,
$\mathbf{V}$, and power allocation matrix $\mathbf{P}$ as
\begin{equation} \label{equ:P2_minPower}
\mathop {\min }\limits_{{\bf{U}},{\bf{V}},{\bf{P}}} \sum\limits_{k
= 1}^K {\sum\limits_{j = 1}^{L_k } {p_{kj} } } {\rm{\ \ subj.\ to\
}}\mathop {{\rm{min}}}\limits_k \frac{{{\rm{\overline{SINR}}}_k
}}{{\gamma _k }} \ge 1 {\;\; \mbox{and} \;\; \sum\limits_{k
= 1}^K {\sum\limits_{j = 1}^{L_k } {p_{kj} } } \leq P_{\max}}. \\
\end{equation}

\subsection{Iterative methods based on uplink-downlink duality}\label{sec:duality}
We briefly review the uplink-downlink duality, which plays an
important role in finding efficient solutions based on iterative
methods for our problems. In \cite{Tse02, Boche03, Schubert02,
Jindal04}, it was shown that it is always possible to find a
virtual uplink system for the downlink system. We plot the virtual
uplink for user $k$ in the lower part of Fig. \ref{fig:DL_UL},
where $\mathbf{Q}_k$ is the corresponding power allocation matrix
in the virtual uplink defined similarly as $\mathbf{P}_k$. To be
more specific, with fixed beamforming filters $\bf U$, $\bf V$,
SINR targets $\gamma _1, \ldots, \gamma _K$, and the same sum power
constraint $P_{\rm max}$ for both the downlink and the virtual
uplink, the downlink and its virtual uplink system have the same
SINR to target ratio with optimal $\bf P$ and $\bf Q$.

With the aids of the uplink-downlink duality, the optimization
problems Pr and Pp in Section \ref{sec:problem} can be solved
efficiently with iterative algorithms. Now we introduce the basic
concepts of these algorithms, as summarized in Table
\ref{tab:AlgBasic}. For simplicity, we use Problem Pr as an
example. From Table \ref{tab:AlgBasic}, for iteration $n$, with
the downlink transmitter and receiver beamformers ${\bf U}^{(n)}$
and ${\bf V}^{(n)}$ fixed, we can obtain a new power allocation
matrix ${\bf P}^{(2n+1)}$ to increase the minimum SINR to target
ratio $\mathop {{\rm{min}}}\limits_k
{\rm{\overline{SINR}}}_k/\gamma_k$. Note that the downlink power
allocation are executed two times (Step 1 and 3) for the $n$th
iteration, as shown in Table \ref{tab:AlgBasic}. To simplify
notations in the following sections, we use ${\bf P}^{(2n+1)}$ and
${\bf P}^{(2n+2)}$ to represent the new power allocation matrices
for the first and second downlink power allocations respectively.
With fixed ${\bf P}^{(2n+1)}$ and ${\bf U}^{(n)}$, we can obtain a
new downlink receiver beamformer ${\bf V}^{(n+1)}$ to increase
${\rm{\overline{SINR}}}_k/\gamma_k$ for all users. The minimum
ratio $\mathop {{\rm{min}}}\limits_k
{\rm{\overline{SINR}}}_k/\gamma_k$ is further optimized using the
new power allocation matrix ${\bf P}^{(2n+2)}$ computed from ${\bf
U}^{(n)}$ and ${\bf V}^{(n+1)}$.  Then we turn to the virtual
uplink to update ${\bf U}^{(n)}$. Similarly, fixing uplink
transmitter beamformer ${\bf V}^{(n+1)}$ and receiver beamformer
${\bf U}^{(n)}$, we obtain a new uplink power allocation matrix
${\bf Q}^{(2n+1)}$. After power allocation, the SINR to target
ratios of the downlink and virtual uplink are equal. Then we can
find ${\bf U}^{(n+1)}$ based on ${\bf Q}^{(2n+1)}$ and ${\bf
V}^{(n+1)}$. After that, ${\bf Q}^{(2n+2)}$ is updated according
to the new ${\bf U}^{(n+1)}$ and ${\bf V}^{(n+1)}$, and so on.

Note that all the iterations are done at the transmitter,
and the transmitter does not need to feed forward the optimized
receive filters to the receivers during the iterations. The
receiver can compute the final filter by itself after the
iterative algorithm stops. This procedure is the same as
\cite[Section II-B]{CaireEstimation}, and we briefly describe it
here. First, as in the ``common training'' phase in
\cite{CaireEstimation}\cite{biguesh2006training}, each receiver
$k$ can estimate its own channel ${\bf{H}}_k$ by using the known
training sequence. After receiver $k$ feeds back ${\bf{H}}_k$ to
the transmitter, the transmitter can iteratively compute transmit
and receive beamforming filters, as well as power allocation
matrices in Table \ref{tab:AlgBasic} according to ${\bf{H}}_k$.
After the iterative algorithm stops, the ``dedicated training'' phase as in
\cite{CaireEstimation} is performed to let the receivers compute
the final receiver filter. In this phase, the transmitter will
broadcast orthogonal training sequences to the receivers as in
\cite{CaireEstimation}, and each receiver can estimate the final
equivalent channel formed by ${\bf{H}}_k$, the transmit filters, and power
allocation matrices to calculate its final receive beamformer. We
will first show how to calculate the beamforming filters in the
next section, and then show how to use these filters to determine
power allocation in Sections \ref{sec:GroupPowerAlloc} and
\ref{sec:PerStreamPowerAlloc}.

\section{ Group Maximum SINR Filter Bank for the Average SINR
Constraint}\label{sec:GSINR-FB} In this section, we
introduce the key motivation of our paper, that is, the use of
GSINR-FB in \cite{HJSu02} as the beamfomer to solve
\eqref{equ:P1_sumPowerConstraint} \eqref{equ:P2_minPower}. This
filter bank is a non-trivial generalization of the one used in
\cite{kha06}. It uses the dimensions provided by the multiple
receive antennas at each user more efficiently than \cite{kha06}.
Specifically, the streams of each user (or group) cooperate with one
another in our scheme, rather than interfere with one another as in
\cite{kha06}. Since this filter bank maximizes the total SINR of the
streams of each user, it also maximizes the average SINR criterion
adopted in this paper. We will also review the applications of
the average SINR criterion at the end of this section.

\subsection{ Group Maximum SINR Filter Bank}\label{chap:MSINR}
To solve \eqref{equ:P1_sumPowerConstraint}
\eqref{equ:P2_minPower}, the GSINR-FB is adopted for our
transmitter beamformer $\bf U$ and receiver beamformer $\bf V$ to
maximize the average SINR. Moreover, as will be shown in
Proposition \ref{Prop_SINR_b}, the optimal SINR balancing
structure based on the GSINR-FB beamforming will make the
corresponding power allocation problem trackable. Let us first focus on Step 2 in Table \ref{tab:AlgBasic}, that
is, given ${\bf U}^{(n)}$ and ${\bf P}^{(2n+1)}$, finding filter
${\bf V}^{(n+1)}$ to maximize
${\sum\limits_{j=1}^{L_k}{\mathrm{SINR}_{kj}^\mathrm{DL}}},
\forall k$ ($L_k$ times of the average SINR), where
${\mathrm{SINR}_{kj}^\mathrm{DL}}$ is the SINR of the $j$th stream
of user $k$ in this step. For brevity, we shall omit the iteration
index $n$ in most of the following equations. Following
\cite{HJSu02}, the optimization problem becomes
\begin{equation} \label{Eq_MSNR}
\mathop {\max } \limits_{\mathbf{V}_k}
{\sum\limits_{j=1}^{L_k}{{\bf{v}}_{kj}^H
{\bf{R}}_{s,k}^{{\rm{DL}}} {\bf{v}}_{kj} },} {\mbox{\ \ subj. to\ \ }}
{{\bf{v}}_{kj}^H {\bf{R}}_{n,k}^{{\rm{DL}}} {\bf{v}}_{kj}=1 }, \;\;
\forall j,
\end{equation}
where $\mathbf{V}_k = [\mathbf{v}_{k1},\ldots,\mathbf{v}_{kL_k}]$,
while
\begin{equation} \label{eq_MSNR_AB}
{\bf{R}}_{s,k}^{{\rm{DL}}} = {\bf{H}}_k^H {\bf{U}}_k {\bf{P}}_k
{\bf{U}}_k^H {\bf{H}}_k \;\; \mbox{and} \;\;
{\bf{R}}_{n,k}^{{\rm{DL}}} = \sum\limits_{i \ne k} {{\bf{H}}_k^H
{\bf{U}}_i {\bf{P}}_i {\bf{U}}_i^H {\bf{H}}_k }  + \sigma ^2
{\bf{I}}_{N_k },
\end{equation}
are the signal covariance matrix and the interference-plus-noise
covariance matrix for \textit{user $k$}, respectively. It is now
evident that we must let $L_k \le \min\left\{M, N_k\right\}$ since
the number of eigenvectors is limited by the dimension of ${\bf
H}_k$. The optimization problem in (\ref{Eq_MSNR}) was shown to be
equivalent to solving the generalized eigenvalue problems
\cite{HJSu02} as
\begin{equation}
\label{equ:DLeigenproblem}
\mathbf{R}_{s,k}^\mathrm{DL}\mathbf{v}_{kj} =
\lambda_{kj}^\mathrm{DL}\mathbf{R}_{n,k}^\mathrm{DL}\mathbf{v}_{kj},
\;\; \forall j
\end{equation}
with
\begin{equation}
\lambda _{kj}^{{\rm{DL}}}  = \frac{{{\bf{v}}_{kj}^H
{\bf{R}}_{s,k}^{{\rm{DL}}} {\bf{v}}_{kj} }}{{{\bf{v}}_{kj}^H
{\bf{R}}_{n,k}^{{\rm{DL}}} {\bf{v}}_{kj} }} =
{\rm{SINR}}_{kj}^{{\rm{DL}}}.
\end{equation}
Then ${\bf V}_k$ can be computed easily. The receive beamforming
filter designed for the downlink can be carried over to the
transmit beamforming filter for uplink, and vice versa. Thus the
receive beamforming filter ${{\bf U}^{(n+1)}}$ for the virtual
uplink system in Step 4 in Table \ref{tab:AlgBasic} can be
computed similarly.

Now we show why the GSNIR-FB performs better than those in
\cite{Sch04}\cite{kha06}. In \cite{kha06}, all streams interfere
with one another and $\mathbf{v}_{kj}$ satisfies
\[
\mathbf{R}_{s,kj}^\mathrm{DL}\mathbf{v}_{kj} =
\lambda_{M,kj}^\mathrm{DL}\mathbf{R}_{n,kj}^\mathrm{DL}\mathbf{v}_{kj},
\]
where $\lambda_{M,kj}^\mathrm{DL}$ is the maximum generalized
eigenvalue of
($\mathbf{R}_{s,kj}^\mathrm{DL},\mathbf{R}_{n,kj}^\mathrm{DL}$);
\begin{equation} \label{eq_Kha_AB}
\mathbf{R}_{s,kj}=\bH^H_k \bu_{kj}\bu^H_{kj} \bH_k \; \mbox{and}
\; \; {\bf{R}}_{n,kj}^{{\rm{DL}}}  = \sum_{\ell=1,\ell \neq
j}^{L_k} p_{k \ell}\bH^H_k \bu_{k\ell}\bu^H_{k\ell}
\bH_k+\sum\limits_{i \ne k} {{\bf{H}}_k^H {\bf{U}}_i {\bf{P}}_i
{\bf{U}}_i^H {\bf{H}}_k } + \sigma ^2 {\bf{I}}_{N_k }
\end{equation}
are the signal covariance matrix and the interference-plus-noise
covariance matrix for \textit{stream $j$ of user $k$},
respectively, and
$\mathbf{U}_k = [\mathbf{u}_{k1},\ldots,\mathbf{u}_{kL_k}]$. Comparing
\eqref{eq_Kha_AB} with \eqref{eq_MSNR_AB}, one can easily see that, in \cite{kha06},
the streams of the same user interfere with one another and there
is additional intra-group interference in
${\bf{R}}_{n,kj}^{{\rm{DL}}}$ (the first term of
${\bf{R}}_{n,kj}^{{\rm{DL}}}$) compared with
${\bf{R}}_{n,k}^{{\rm{DL}}}$ in \eqref{eq_MSNR_AB}. The GSINR-FB
beamforming exploits additional dimensions from the multiple
receiver antennas, which are not provided in \cite{Sch04} (where $N_k$=1), much more efficiently, by letting the streams of each
user cooperate rather than compete as in \cite{kha06}.

\subsection{  Average SINR criterion and its applications}
\label{secASINRapp} The average SINR criterion $\overline{\rm
SINR}_k$ is very useful in many communication systems
\cite{JHWang08,Davision2006robust,HJSu02} and can serve as a good
metric for the QoS. Here we briefly review some of its
applications. Note that in these applications, it is the total
SINR $L_k{\overline{\rm SINR}}_k$ which serves as the performance
metric, which equals to $L_k$ times the average SINR. However,
as will be discussed in Section \ref{sec:SimulationResult}, to have a fair comparison with the results in \cite{kha06} where the per stream
SINR is considered, the average SINR is used in the comparison.
\\

\noindent \textbf{Approximation of maximum achievable rate at low
SINR \cite{JHWang08}:} The maximum achievable rate for user $k$
is
\begin{equation}
\begin{array}{l}
\hspace{4mm}  \sum_{j=1}^{L_k}{\rm log}(1+\frac{{\rm SINR}_{kj}}{\Gamma})\\
=  {\rm log}\prod_{j=1}^{L_k}{(1+\frac{{\rm SINR}_{kj}}{\Gamma})}\\
\approx {\rm log}(1+{L_k\overline{\rm SINR}}_k/\Gamma),
\end{array}
\end{equation}
where $\Gamma$ is the SNR gap to capacity \cite[P.432]{Haykin01}
\cite[Chapter 7]{proakis} due to suboptimal channel coding schemes
and the limitation of circuit implementation in practical systems.
According to \cite[P.432]{Haykin01}, the gap is huge (8.8 dB) for uncoded
PAM or QAM operating at $10^{-6}$ bit error rate. This
approximation is also useful in systems with large numbers of users where the total
interference power in
\eqref{equ:DL_AfterV} is large. \\

\noindent \textbf{Receiver SINR
\cite{Davision2006robust}\cite{HJSu02}:} Assuming that the maximum
ratio combining (MRC) is applied to $\hat{\bx}_k$ in
\eqref{equ:DL_AfterV}, the receiver SINR at the output of the MRC
is the sum of individual SINRs as $L_k{\overline{\rm SINR}}_k$.
This metric is very useful when space-time coding is applied
and ${\bx}_k$ contains the space-time coded symbols. In this case, the decoding is based on the MRC results \cite{HJSu02}. \\

\noindent \textbf{Minimization of the pairwise error probability
\cite{JHWang08}:} When a space-time block code (STBC) is applied
and ${\bx}_k$ contains the STBC symbols. Assuming that the
channel is slow fading and remains constant during the
transmission of a codeword, and that the maximum-likelihood
detector is used at the receiver, one can approximately
transform the minimization of the pairwise codeword error probability
to the maximization of $L_k{\overline{\rm SINR}}_k$ following the steps in
\cite{JHWang08}. This approximation applies to both the orthogonal and quasi-orthogonal STBCs.

\section{Power Allocation}\label{sec:PowerAlloc}
Now we focus on the optimal power allocation strategy for the Step
3 in Table \ref{tab:AlgBasic}, where the maximum SINR beamforming
filter banks ${\bf U}^{(n)}$, ${\bf V}^{(n+1)}$ and a set of SINR
targets $\gamma_1, \ldots, \gamma_K$ are given. The optimization
problem corresponding to Problem Pr
(\ref{equ:P1_sumPowerConstraint}) is
\begin{equation}
\mathop {\max }\limits_{\bf{P}} \mathop {\min }\limits_k
\frac{{{\overline{\rm{SINR}}}_k^{{\rm{DL}}} }}{{\gamma _k }}{\rm{\
\ subj.\ to\ \ }} \sum\limits_{k = 1}^K {\sum\limits_{j = 1}^{L_k
} {p_{kj} } } \leq P_{\max }. \label{equ:sumpowerconstraint}
\end{equation}
The other one corresponding to Problem Pp
(\ref{equ:P2_minPower}) which minimizes the total transmitted
power, such that each individual SINR target can be achieved, is
\begin{equation}
\label{equ:powerMinimizationProblem} \mathop {\min
}\limits_{\bf{P}} \sum\limits_{k = 1}^K {\sum\limits_{j = 1}^{L_k
} {p_{kj} } } {\rm{\ \ subj.\ to\ \ }} \mathop {\min }\limits_k
\frac{{{\overline{\rm{SINR}}}_k^{\rm DL} }}{{\gamma _k }} \ge
1{\rm{,\ }} \;\; \mbox{and} \;\; \sum\limits_{k = 1}^K
{\sum\limits_{j = 1}^{L_k } {p_{kj} } } \leq P_{\max}.
\end{equation}

We will first explore the structure of the optimal
solutions for these problems in Section \ref{sec:SINR_balan}.
However, even with this structure which significantly simplifies
the problems, the two per-steam power allocation problems are
very complicated and the solutions in \cite{Sch04}\cite{kha06} do
not apply. Thus, we first intensionally introduce some restrictions to
the power allocation strategies to simplify the problems and
benefit from the simple power allocation schemes similar to those
in \cite{Sch04}\cite{kha06}. In Section
\ref{sec:SimulationResult}, the simulation results show that even
without the new per stream power allocation, the performance of
\cite{Sch04}\cite{kha06} can be enhanced by simply applying
the GSINR-FB as the beamformers. This verifies our motivation to use
the GSINR-FB. The results for the simple ``grouped'' power
allocation are presented in Section \ref{sec:GroupPowerAlloc}. We
then remove the restrictions and present the general per-stream
power allocation results in Section \ref{sec:PerStreamPowerAlloc}.
The insights to why the proposed algorithms perform better
than those in \cite{Sch04}\cite{kha06} are given in Section
\ref{sec:Insight}.

\subsection{ Optimal SINR balancing structure under GSINR-FB beamforming}
\label{sec:SINR_balan} By carefully rearranging the complicated
$\overline{\rm{SINR}}_k^{\rm DL}$ to a simpler equivalent form and
using the properties of the GSINR-FB, we prove the following
structure for the optimal power allocation which makes solving the
complicated power allocation problems
\eqref{equ:sumpowerconstraint}
\eqref{equ:powerMinimizationProblem} possible.
\\
\begin{proposition} \label{Prop_SINR_b}
For the optimization problem (\ref{equ:sumpowerconstraint}), the
optimal solution $\bf P$ makes all users achieve the same SINR
to target ratio, that is, $ {\overline{\rm SINR}}_k^{\rm DL} /
\gamma_k = {C}^{\rm DL}$, for all $k$. Here ${C}^{\rm DL}$ is the
SINR balanced level.\\
\end{proposition}


\begin{proof}
The vector norms of the beamforming filters $\mathbf{v}_{kj}$,
$j=1...L_k$, can be adjusted such that

\begin{itemize}
    \item[1)] ${{\bf{V}}_{k}^H {\bf{R}}_{n,k}^{{\rm{DL}}} {\bf{V}}_{k} }$ is a
scaled identity matrix \cite{HJSu02},
    \item[2)] ${\rm trace}\left( {{\bf{V}}_k^H {\bf{V}}_k } \right) =
L_k$.
\end{itemize}
When the above two conditions are satisfied, the average SINR
of user $k$ in the downlink scenario can be expressed as
\begin{equation}
\label{equ:avgSINR_DL} \overline {{\rm{SINR}}} _k^{{\rm{DL}}} =
\frac{1}{{L_k }}\sum\limits_{j = 1}^{L_k }
{{\rm{SINR}}_{kj}^{{\rm{DL}}} } = \frac{{{\rm{trace}}\left(
{{\bf{V}}_k^H {\bf{R}}_{s,k}^{{\rm{DL}}} {\bf{V}}_k }
\right)}}{{{\rm{trace}}\left( {{\bf{V}}_k^H
{\bf{R}}_{n,k}^{{\rm{DL}}} {\bf{V}}_k } \right)}}.
 \end{equation}
Expanding ${\bf R}_{s,k}$ and ${\bf R}_{n,k}$,
\begin{equation}
\overline {{\rm{SINR}}} _k^{{\rm{DL}}}  = \frac{{{\rm{trace}}\left(
{{\bf{V}}_k^H {\bf{H}}_k^H {\bf{U}}_k {\bf{P}}_k {\bf{U}}_k ^H
{\bf{H}}_k {\bf{V}}_k } \right)}}{{\sum\limits_{j \ne k}
{{\rm{trace}}\left( {{\bf{V}}_k^H {\bf{H}}_k^H {\bf{U}}_j {\bf{P}}_j
{\bf{U}}_j^H {\bf{H}}_k {\bf{V}}_k } \right)}  + L_k \sigma ^2 }}.
\label{equ:SINRtrace}
\end{equation}
Since ${\rm{trace}}\left( {{\bf{XY}}} \right) = {\rm{trace}}\left(
{{\bf{YX}}} \right)$ \cite{HornMatrix}, the $\rm trace(\cdot)$
terms can be written as
\begin{equation}
\begin{array}{l}
\hspace{4mm} {\rm trace}({\bf{V}}_k^H {\bf{H}}_k ^H {\bf{U}}_j {\bf{P}}_j {\bf{U}}_j ^H {\bf{H}}_k {\bf{V}}_k ) \\
 = {\rm trace}({\bf{P}}_j {\bf{U}}_j ^H {\bf{H}}_k {\bf{V}}_k {\bf{V}}_k^H {\bf{H}}_k^H {\bf{U}}_j ) \\
 = \sum\limits_{l = 1}^{L_j } {p_{jl} [{\bf{A}}_{jk} ]_{ll} },\\
\end{array}
\end{equation}
where ${\bf{A}}_{jk}  \buildrel \Delta \over = {\bf{U}}_j ^H
{\bf{H}}_k {\bf{V}}_k {\bf{V}}_k^H {\bf{H}}_k ^H {\bf{U}}_j$ and
$[\mathbf{A}_{jk}]_{ll}$ denotes the $l$th diagonal element of
$\mathbf{A}_{jk}$. Therefore, the average SINR of user $k$ is
\begin{equation}
{\rm{\overline {SINR}}}_k^{{\rm{DL}}} = \frac{{\sum\limits_{l =
1}^{L_k } {p_{kl} [{\bf{A}}_{kk} ]_{ll} } }}{{\sum\limits_{j=1,j\ne
k}^K {\sum\limits_{l = 1}^{L_j} {p_{jl} [{\bf{A}}_{jk} ]_{ll} } } +
L_k \sigma ^2 }}.
    \label{equ:SINRAntennaPowera}
\end{equation}

Observing \eqref{equ:SINRAntennaPowera}, we know that the
maximizer of the optimization problem
(\ref{equ:sumpowerconstraint}) satisfies
\begin{equation}
\label{equ:equalC} \frac{{\overline {{\rm{SINR}}} _k^{{\rm{DL}}}
}}{{\gamma _k }} = {C}^{{\rm{DL}}},\ 1 \le k \le K.
\end{equation}
The reason is as the following. Since
$[\mathbf{A}_{jk}]_{ll}>0, \; \forall j,k,l $, each $\overline
{{\rm{SINR}}} _k^{{\rm{DL}}}$ is strictly monotonically increasing
in $p_{kl}$ and monotonically decreasing in $p_{jl}$ for $j \ne
k$. Thus all users must have the same SINR to target ratio
${C}^{{\rm{DL}}}$. Otherwise, the users with higher SINR to target
ratios can give some of their power to the user with the lowest
ratio to increase it, which contradicts the optimality. \\
\end{proof}
Following the same steps of the above proof, the SINR balancing
structure also exists for Problem Pp in
(\ref{equ:powerMinimizationProblem}). Now we can solve power
allocation problems (\ref{equ:sumpowerconstraint}) and
(\ref{equ:powerMinimizationProblem}) with the aid of Proposition
\ref{Prop_SINR_b} which makes these problem trackable as shown in
the following.

\subsection{Simplified Solution: Group Power Allocation}\label{sec:GroupPowerAlloc}
For clarity, we present the simple group power allocation first
then the general per-stream power allocation in the next
subsection. The group power allocation intentionally restricts the
power allocation strategy to make the complicated power allocation
problem with multiple receiver antennas similar to the
simple one in \cite{Yang98}\cite{Sch04} where $N_k=1$. Thus the
group power allocation takes the advantage of the spatial
diversity provided by the GSINR-FB based beamforming to improve
the performance, while keeping the complexity moderate.

To be more specific, the allocated power for a user using the group
power allocation is evenly distributed over all streams of that
user as
\begin{equation}
p_{k1}=p_{k2}=...=p_{kL_k},\ 1\le k \le K.
\end{equation}
Let the power allocated to user $k$ be $p_k$. Consequently, the
diagonal power allocation matrix $\mathbf{P}_k$ for user $k$ can
be written as a scaled identity matrix, that is,
\begin{equation} \label{eq_group_cons}
\mathbf{P}_k=\frac{p_k}{L_k} \mathbf{I}_{L_k}.
\end{equation}
We also define a vector $\mathbf{p}=[p_1,\ldots,p_K]^{\mathrm{T}}$
to replace matrix $\bf P$ in the optimization problems.
Substituting $\mathbf{P}_k=\frac{p_k}{L_k} \mathbf{I}_{L_k}$ into
Equation (\ref{equ:avgSINR_DL}), the average SINR in problems
(\ref{equ:sumpowerconstraint}) and
(\ref{equ:powerMinimizationProblem}) is
\begin{equation}
\overline {{\rm{SINR}}} _k^{{\rm{DL}}} = \frac{{\frac{{p_k
}}{{L_k^2 }}\left\| {{\bf{V}}_k^H {\bf{H}}_k^H {\bf{U}}_k }
\right\|_F^2 }}{{\sum\limits_{j \ne k} {\frac{{p_j }}{{L_j L_k
}}\left\| {{\bf{V}}_k^H {\bf{H}}_k^H {\bf{U}}_j } \right\|_F^2 } +
\sigma ^2 }}. \label{equ:SINRDL}
\end{equation}
With the ``grouped'' constraint on the power allocation strategy
\eqref{eq_group_cons}, the simplified average SINR
\eqref{equ:SINRDL} for $N_k>1$ has the same structure as that in
\cite{Yang98}\cite{Sch04} where $N_k=1$. Thus the solutions of
this simplified group power allocation for Problems Pr and Pp can
be easily obtained. These solutions are briefly presented in the
following subsections. The overall optimization algorithms are
also summarized at the end of each subsection.

\textbf{Group Power Allocation for Problem Pr:} With the SINR
balancing structure from the GSINR-FB beamforming in Proposition
\ref{Prop_SINR_b}, the group power allocation for Problem Pr
(\ref{equ:sumpowerconstraint}) can be solved by a simple
eigensystem as
\begin{equation} \label{eq_GpowerallocationPr}
\mathbf{\Upsilon} {\bf{\tilde{p}}} = \frac{1}{{C^{{\rm{DL}}}
}}{\bf{\tilde{p}}},
\end{equation}
where the extended coupling matrix $\mathbf{\Upsilon}$ and the
extended power vector ${\bf{\tilde p}}$ are defined as
\begin{equation}
\mathbf{\Upsilon}  = \left[ {\begin{array}{*{20}c}
   {{\bf{D\Psi }}} & {{\bf{D\sigma }}}  \\
   {\frac{1}{{P_{\max } }}{\bf{1}}^{T} {\bf{D\Psi }}} & {\frac{1}{P_{\rm max}}{\mathbf{1}^{T}\mathbf{D\bm{\sigma} }}}  \\
\end{array}} \right] \;\; \mbox{and} \;\; \left[ {\begin{array}{*{20}c}
   {\bf{p}}  \\
   1  \\
\end{array}} \right],
\end{equation}
respectively, where
\begin{equation}
{\bf{D}} = diag\left\{ {\frac{{L_1^2 \gamma _1 }}{{\left\|
{{\bf{V}}_1^H {\bf{H}}_1^H {\bf{U}}_1 } \right\|_F^2
}},...,\frac{{L_K^2 \gamma _K }}{{\left\| {{\bf{V}}_K^H
{\bf{H}}_K^H {\bf{U}}_K } \right\|_F^2 }}} \right\} \;\;
\end{equation}
and the $ij$th element of the $K \times K$ matrix $\mathbf{\Psi}$
is zero when $j=i$ or ${\frac{{\left\| {{\bf{V}}_i^H {\bf{H}}_i^H
{\bf{U}}_j } \right\|_F^2 }}{{L_i L_j }}}$ when $j \ne i$.

By using Proposition \ref{Prop_SINR_b} and the simplified average
SINR \eqref{equ:SINRDL} in (\ref{equ:sumpowerconstraint}), the
rest of the proof of the previous results is similar to those in
\cite{Yang98}\cite{Sch04} and omitted. With
(\ref{equ:DLeigenproblem}) and (\ref{eq_GpowerallocationPr}), we
summarize the final optimization algorithm for Problem Pr in Table
\ref{tab:AlgGroupSumPower}, which iteratively calculates the
optimal beamforming filter and power allocation vector between the
downlink and the uplink, where \textit{eig} means the generalized
eigenvalue solver. Due to the uplink-downlink duality described in
Section \ref{sec:duality}, it is guaranteed that the uplink
balanced level $C^{{\rm{UL}}}$ equals to the downlink balanced
level $C^{{\rm{DL}}}$.

\textbf{Group Power Allocation for Problem Pp:} Again, with
Proposition 1, the minimizer of
(\ref{equ:powerMinimizationProblem}) satisfies
\begin{equation}
\label{equ:SINRequalGammak} \overline {{\rm{SINR}}} _k^{{\rm{DL}}}
= \gamma _k ,\ 1 \le k \le K.
\end{equation}
Substituting (\ref{equ:SINRequalGammak}) into (\ref{equ:SINRDL}),
the resulting power allocation vector is
\begin{equation}
\label{equ:GroupMinPowerEq}
{\bf{p}} = ({\bf{I}} - {\bf{D\Psi }})^{ - 1} {\bf{D\bm \sigma }}.
\end{equation}
The optimal $\mathbf{q}$ for the virtual uplink can be obtained
similarly. The overall algorithm for Problem Pp is summarized in
Table \ref{tab:AlgGroupMinPower} which iteratively finds the
optimal solution minimizing the required power.

Note that (\ref{equ:GroupMinPowerEq}) does not necessarily have a
solution with nonnegative elements. When there exists at least one
nonnegative power allocation satisfying the target SINR
constraints and total power constraint $P_{\max}$ in
(\ref{equ:powerMinimizationProblem}), we call the system feasible.
Depending on the channel conditions, the total power required to
achieve the target SINRs could be quite large and exceed
$P_{\max}$. For the purpose of studying the effects of the
algorithms on the system feasibility, we use the sum power
allocation algorithm in Table \ref{tab:AlgGroupSumPower} with a
large $P_{\rm max}$ (43 dBm) to check the feasibility as in
\cite{Sch04,kha06}. In checking the feasibility, as soon as the
balanced level becomes larger than 1 (which means that a feasible
solution can be obtained), the algorithm switches to the power
minimization steps. On the other hand, if the balanced level
remains below 1 when the feasibility testing stage ends, the
feasibility test fails and the power minimization algorithm stops.
In practical applications, when the system is infeasible, one must
relax the constraints by reducing the number of users $K$ or
decreasing the target SINR.


\subsection{ General Solution - Per Stream Power
Allocation} \label{sec:PerStreamPowerAlloc} Now we remove the
restriction of evenly distributing power in a group in Section
\ref{sec:GroupPowerAlloc}. The performance is expected to be
further improved since the group power allocation is a subset of
the per stream power allocation. The general power allocation
solutions presented in this subsection are much more complicated
than the results in \cite{Sch04}\cite{kha06}. The overall
optimization algorithms for Problems Pp and Pr are also summarized
at the end of each subsection.

\textbf{Per Stream Power Allocation for Problem Pp:} The power
minimization problem using the result of Proposition
\ref{Prop_SINR_b} becomes
\begin{equation}
\label{equ:antennaMinPower} \mathop {\min }\limits_{\bf{p}}
\sum\limits_{k = 1}^K {\sum\limits_{j = 1}^{L_k } {p_{kj} } }
{\rm{\ \ s.t.\ }} {\sum\limits_{k = 1}^K {\sum\limits_{j = 1}^{L_k
} {p_{kj} } } \leq P_{\max}} \;\; \mbox{and} \;\; \overline {\rm
SINR} _k^{{\rm{DL}}}  = \gamma _k ,1 \le k \le K.
\end{equation}
With the equivalent SINR expression in
(\ref{equ:SINRAntennaPowera}), we will show that
(\ref{equ:antennaMinPower}) can be elegantly recast as a
well-known linear-programming problem.
We first recall that the
average SINR of user $k$ (\ref{equ:SINRAntennaPowera}) is
\begin{equation}
{\rm{\overline {SINR}}}_k^{{\rm{DL}}} = \frac{{\sum\limits_{l =
1}^{L_k } {p_{kl} [{\bf{A}}_{kk} ]_{ll} } }}{{\sum\limits_{j=1,j\ne
k}^K {\sum\limits_{l = 1}^{L_j} {p_{jl} [{\bf{A}}_{jk} ]_{ll} } } +
L_k \sigma ^2 }}.
    \label{equ:SINRAntennaPower}
\end{equation}
Substituting (\ref{equ:SINRAntennaPower}) into
(\ref{equ:antennaMinPower}), the original power minimization problem
turns into a linear programming problem, that is,
\begin{equation}
\label{eq_PerstreamPp}
\begin{array}{l}
 \min {\rm{\ }}{\bf{1}}^T {\bf{p}} \\
 {\rm{s.t.\ \ }}\sum\limits_{l = 1}^{L_k } {p_{kl} [{\bf{A}}_{kk} ]_{ll}/\gamma_k }  - \sum\limits_{j=1,j\ne k}^K {\sum\limits_{l = 1}^{L_j } {p_{jl} [{\bf{A}}_{jk} ]_{ll} } }  = L_k \sigma^2\\
 \hspace{7mm}{\rm{for\ }}k = 1,...,K, {\rm{\ and\ }} \mathbf{p} \ge_e 0,\\
 \end{array}
\end{equation}
where $\bf p$ represents the vector comprising the
diagonal elements of $\bf P$ as in Section
\ref{sec:GroupPowerAlloc}. It is known that a linear
programming problem can be solved in polynomial time using, for example,
the ellipsoid method or the interior point method
\cite{Boyd03}.

Table \ref{tab:AlgAntennaMinPower} summarizes the proposed
iterative algorithm with group maximum SINR beamforming and per
stream power allocation. The virtual uplink power allocation
problem can be similarly solved as (\ref{eq_PerstreamPp}) with
$\mathbf{A}_{jk}$ replaced by ${\bf{B}}_{jk} \buildrel \Delta
\over = {\bf{V}}_j ^H {\bf{H}}_j ^H {\bf{U}}_k {\bf{U}}_k^H
{\bf{H}}_j {\bf{V}}_j$. Like the group power minimization algorithm
in Table \ref{tab:AlgGroupMinPower}, the feasibility of this algorithm should also
be checked using the per stream sum power allocation which is
described in the next subsection.

\textbf{Per Stream Power Allocation for Problem Pr :} With a fixed
beamforming matrix $\mathbf{U}$, a fixed receive filter
$\mathbf{V}$, and a total power constraint, the optimization
problem obtained by applying Proposition \ref{Prop_SINR_b} in
\eqref{equ:sumpowerconstraint} is
\begin{equation}
\label{equ:PerStreamSumPower}
\begin{array}{l}
 \mathop {\max\ }\limits_{\bf{p}} C^{\rm{DL}} \\
 {\rm{s.t.\ \ }}C^{\rm{DL}} = \frac{{{\rm{SINR}}_k^{{\rm{DL}}} }}{{\gamma _k }}{\rm{,\ \ }}k = 1,...,K\\
 \hspace{7mm}{\rm{and\ \ }}\sum\limits_{k = 1}^K {\sum\limits_{j = 1}^{L_k } {p_{kj} } }  = P_{\max }, \\
\end{array}
\end{equation}
where $\rm{\overline {SINR}}_k^{\rm{DL}}$ is rearranged in form
(\ref{equ:SINRAntennaPower}).

The optimal power allocation vector for this complicated problem
is difficult to obtain, thus we consider a suboptimal solution which can be
found by simple iterative algorithms. First, using the concept of
waterfilling, we fix the proportion of the power of data streams
in each group according to the equivalent channel gains. That is,
let
\begin{equation}
\begin{array}{r}
p_{k1} :p_{k2} : \ldots :p_{kL_k } = [{\bf{A}}_{kk} ]_{11}:[{\bf{A}}_{kk} ]_{22} : \ldots :[{\bf{A}}_{kk} ]_{L_k L_k },\\
\mbox{for } k=1, \ldots, K.
\end{array}
\end{equation}
Therefore, the $L_k$ variables $p_{k1}, \ldots, p_{kL_k}$ can be reduced
to one variable $t_k$ such that $p_{kl}  = t_k \left[
{{\bf{A}}_{kk} } \right]_{ll} /\sum\limits_{i = 1}^{L_k } {\left[
{{\bf{A}}_{kk} } \right]_{ii} }$ for each $l$,
and $\sum\limits_{l = 1}^{L_k } {p_{kl} }
= t_k$. The SINR for user $k$ in Equation
(\ref{equ:SINRAntennaPower}) can be rewritten as
\begin{align}
\label{equ:PerStreamSINRProportion1}
{\overline{\rm{SINR}}}_k^{{\rm{DL}}} & = \frac{{t_k \left(
{\sum\limits_{l = 1}^{L_k } {\left[ {{\bf{A}}_{kk} } \right]_{ll}^2 }
/\sum\limits_{i = 1}^{L_k } {\left[ {{\bf{A}}_{kk} } \right]_{ii} } }
\right)}}{{\sum\limits_{j \ne k} {t_j \left( {\sum\limits_{l =
1}^{L_j } {\left[ {{\bf{A}}_{jj} } \right]_{ll}} \left[
{{\bf{A}}_{jk} } \right]_{ll} /\sum\limits_{i = 1}^{L_j } {\left[
{{\bf{A}}_{jj} } \right]_{ii} } } \right) +
L_k \sigma ^2 } }} \\
\label{equ:PerStreamSINRProportion2}
& = \frac{{t_k g_{kk} }}{{\sum\limits_{j \ne k} {t_j g_{jk} }  + L_k
\sigma ^2 }},
\end{align}
with $\sum\limits_{k = 1}^K {t_k }  = \sum\limits_{k = 1}^K
{\sum\limits_{l = 1}^{L_k } {p_{kl} } }  = P_{\max }$.

We then solve problem (\ref{equ:PerStreamSumPower}) with concepts
similar to the sum power iterative water-filling algorithm
proposed in \cite{Jindal05}. The $n$th iteration of the algorithm
is described in the following. Note that this problem has a
similar form as (\ref{equ:sumpowerconstraint}), thus the balanced
levels, defined as ${\rm{SINR}}_k / \gamma_k$, of all users must
be equal according to Proposition \ref{Prop_SINR_b}. At each
iteration step, we generate a new effective level gain for each
user based on the power of other users from the previous step
$t^o_j,\; j \neq k$ as
\begin{equation}
G_k  = \frac{{g_{kk}/\gamma_k }}{{\sum\limits_{j \ne k} {t_j^{o}
g_{jk} }  + L_k \sigma ^2 }},
\end{equation}
for $k=1, \ldots, K$. The $K$ power variables $t_k$s are
simultaneously updated subject to a sum power constraint. In order
to maintain an equal level, we allocate the new power
proportionally to the inverse of the level gain of each user as
\begin{equation}
t_k  = \frac{{P_{\max } }}{{G_k} \sum\limits_{j = 1}^K
{\frac{1}{{G_j}}} }.
\end{equation}
Note that when updating $t_k$, the power variables of other users are
treated as constants and $t_k>0$.

Similarly, for the virtual uplink, we denote the power variable
for user $k$ as $s_k$ and the effective level gain for user $k$ as $H_k$.
The proposed algorithm for the overall problem Pr is summarized in
Table \ref{tab:AlgAntennaSumPower}.

\subsection{Insights to the performance advantage of the proposed
approaches} \label{sec:Insight} The insights to why the proposed
approaches outperform those in
\cite{kha06}\cite{Sch04} are discussed as follows. First, under the
same power allocation matrix $\mathbf{P}$ in
\eqref{equ:P1_sumPowerConstraint} and \eqref{equ:P2_minPower}, the
GSINR-FB will perform better than the beamformers in \cite{kha06}.
This is because the streams of each user cooperate
with one another in our scheme rather than interfere with one
another as in \cite{kha06}. The mathematical validation was given in
Section \ref{chap:MSINR}. Indeed, as shown in \cite[Section
III]{HJSu02}, the GSINR-FB includes the minimum mean-squared error
(MMSE) filter used in \cite{Sch04}\cite{kha06} as a special case (without cooperation).
Thus the GSINR-FB should have a better performance. As for the power allocation
part, note that our sub-optimal group power allocation has a formulation similar
to that of the power allocation methods in \cite{Sch04}\cite{kha06}. Thus they should be similar in terms of optimality. Our
more complicated per-stream power allocation includes the
group power allocation as special case. Therefore it should perform
better than the group power allocation and the power allocation methods in \cite{Sch04}\cite{kha06}.

Finally, we note that we have no proof whether our iterative algorithms converge to the global optimum or merely local optima. However, as shown by the simulation in the next section, the local optima still result in
much better performance than \cite{Sch04}\cite{kha06}.

\section{ Simulation Results}\label{sec:SimulationResult}
In this section we provide some numerical results to illustrate
the advantages of the proposed algorithms over \cite{kha06} and
the simple BD methods \cite{Choi04}. The design concept of the BD
transmit beamformer is to remove the inter-user interference
in \eqref{equ:DL_ReceiveVector} completely. A BD beamformer can
be found when $M>\sum^K_{i=1,i \ne k} N_i, \forall k$. To
solve Problems Pr and Pp in \eqref{equ:P1_sumPowerConstraint} and
\eqref{equ:P2_minPower}, respectively, and to maximize the average SINR of
the worst user, we also apply the GSINR-FB
as the receive beamformers for the BD cases. Note that this paper focuses on the QoS of
individual users, where the average SINR serves as a metric of QoS. For the BD cases, the
conventional BD receive beamformer design is more for the purpose of sum rate maximization (with waterfilling power allocation), which usually does not maximize the SINR of the worst user.
Thanks to Proposition \ref{Prop_SINR_b}, the
corresponding power allocations can be derived similarly to those
in Section \ref{sec:GroupPowerAlloc} and the details are omitted
here. We also consider both the group and per stream power
allocation strategies for BD, named ``group BD" and ``per stream BD",
respectively.

For the system simulation parameters, the channel matrix
${\mathbf{H}}^H$ is assumed flat Rayleigh faded with independent
and identically distributed (\emph{i.i.d.}) complex Gaussian
elements with zero mean and unit variance. The noise is white
Gaussian with variance 1 W. The transmitter is assumed to
have perfect knowledge of the channel matrix ${\mathbf{H}}^H$, and
each user knows its own equivalent channels as discussed in Section
\ref{sec:duality}. Since typically the transmitter has more
antennas than the receivers, we set the number of streams $L_k$
equal to the number of receive antennas $N_k$ for user $k$. Without
loss of generality, we assume a common SINR constraint
$\gamma$ for all users, i.e., $\gamma
_k = \gamma$ for all $k$. In the following simulation, we generate
1000 channel realizations and average the performance. The
convergence criterion $\epsilon$ of the iterative algorithms is
set to $10^{-3}$.

Fig. \ref{fig:SumPowerNoCombine} shows the simulation results of
the balanced level ${C}^{\rm DL}$ versus total power $P_{\rm max}$
for Problem Pr, where ${C}^{\rm DL}$ is defined as in Proposition
\ref{Prop_SINR_b}. The two proposed algorithms in Table
\ref{tab:AlgGroupSumPower} and Table \ref{tab:AlgAntennaSumPower}
are compared with the method proposed in \cite{kha06} and BD. Note
that in \cite{kha06}, the data streams are processed
separately and the balanced levels are the same for all streams.
With a common SINR constraint $\gamma$ to be satisfied by all
streams, the per-stream balanced level defined in \cite{kha06}
gives the same value as the balanced level ${C}^{\rm DL}$ defined
in Proposition \ref{Prop_SINR_b}. So the comparison of ${C}^{\rm
DL}$ is fair in Fig. \ref{fig:SumPowerNoCombine}. The simulation
parameters are: $K=4$ users, $M=8$ transmit antennas, each user
has 2 receive antennas and 2 streams ($N_k=L_k=2, \;\forall k$),
and the SINR constraint $\gamma=1$. For each channel realization,
all the three algorithms run until convergence but for at most 50
iterations. For fair comparison, only the cases where all the
three algorithms have converged within 50 iterations are
considered in averaging the performance. We will discuss the
convergence probabilities later. It can be seen that the proposed
group power allocation achieves higher balanced levels than the
method in \cite{kha06} at the positive SINR region. The proposed
per stream power allocation further outperforms group power
allocation. Similarly, the per stream BD achieves higher balanced
levels than the group BD since the group BD is a special case of
the per stream BD. Note that the BD schemes perform better when the
total available power $P_{\max}$ is high and perform worse when
the available power is low, since BD is a zero-forcing method
which suffers from the noise enhancement problem at low
$P_{\max}$. When extremely large power is available, BD will
perform close to the proposed methods. However, the operating
region where this phenomenon is obvious needs a much higher power than
our setting in Fig. \ref{fig:SumPowerNoCombine}. We do not
show the simulation results in this region since it is less
practical.

In Fig. \ref{fig:MinPowerNoCombine}, we plot the minimum total
required power $P_{\rm min}$ versus SINR constraint $\gamma$ for
Problem Pp. Simulation parameters are $K=2$ users, $M=8$ transmit
antennas, and each user has $N_k=4$ receive antennas and
$L_k=4$ streams. Again, for the method in \cite{kha06}, a common SINR
target $\gamma$ has to be achieved by all streams. Thus it has the
same average SINR target $\gamma$ for each user as the other
algorithms. For each channel realization, all algorithms first
perform feasibility test using a large $P_{\rm max} = 43$ dBm.
Feasibility test for the method in \cite{kha06} can be done
similarly as the proposed algorithms. As soon as the feasibility
test passes, the corresponding algorithm switches to the power
minimization steps and runs until convergence but for at most 50
iterations. Feasibility test for BD can be done trivially. Only
the cases where all the algorithms have passed the
feasibility test, and converged within 50 iterations, are
considered in averaging the performance. Again, we will defer the
discussions for the infeasible cases and the convergence issues
later. As shown in the figure, the proposed group power allocation
performs better than the method in \cite{kha06} at high SINR.
However, at low SINR it requires more power. This is because group
power allocation suffers for the fact that it cannot adjust the
power within a group as the method in \cite{kha06}. At low SINR,
the interference is larger and the method in \cite{kha06} can
adjust the power within a group to better deal with the
interference. On the other hand, the proposed per stream power
allocation performs better than the other algorithms at both high and low
SINR. Similar to Fig. \ref{fig:SumPowerNoCombine}, the BD methods
perform better than the method in \cite{kha06} at high SINR, but are worse
than the proposed methods in all cases presented. The results at
extremely high power, where the performances of BD and the
proposed methods are close, are not shown due to the same reason
discussed before.

We also present the sum rate comparison in Fig.
\ref{fig:SumRate} where the balanced levels of all users are the
same (as in Fig. \ref{fig:SumPowerNoCombine}) as an indication of the QoS guaranteed and the fairness achieved. The simulation parameters are the same as in Fig. \ref{fig:SumPowerNoCombine}. Note
that the sum rates of both BD methods are worse than the method in \cite{kha06},
while their balanced levels cross over that of \cite{kha06}
in Fig. \ref{fig:SumPowerNoCombine}. This is because under the same average SINR, the method in \cite{kha06} will
make all streams of a user have equal SINR and achieve the
highest sum rate due to the concavity of the $\log$ function. Thus
when a scheme's balanced level advantage over the method in \cite{kha06} is
not significant enough (e.g., the BD schemes), its sum rate may be
lower than that of \cite{kha06}. We emphasis
again that our algorithms focus on the QoS (average SINR) of individual users. Our problem formulations
are fundamentally different from those focusing on sum rate
optimization and not guaranteeing the QoS.

Now we show the feasibility and convergence properties of the proposed
algorithms. In the above simulation setting, the number
of transmit antennas $M$ is equal to the total number of data
streams of all users $\sum^K_{k=1}L_k$ (also equal to the total number
of receive antennas $\sum^K_{k=1}N_k$). We further consider the cases
where $M<\sum^K_{k=1}L_k$ by increasing the number of users $K$.
That is, for Problem Pp, $K=3$, $M=8$, and $N_k=L_k=4, \forall
k;$ while for Problem Pr, $K=5$, $M=8$, $N_k=L_k=2, \forall k$.
We name these cases as Case 2 and the settings for
Fig. \ref{fig:SumPowerNoCombine} and \ref{fig:MinPowerNoCombine}
as Case 1. Note that typically the system will perform scheduling
\cite{Book_Tse} when $M<\sum^K_{k=1}L_k$, that is, it uses
time-division multiple access (TDMA) to schedule a number of users
such that $M=\sum^K_{k=1}L_k$ each time. Thus the simulation
results of Case 1 represent the performance of fully loaded
systems and those of Case 2 well represent the performance of
over-loaded systems.

First we discuss the feasibility issues. From the simulations of
Case 1, we observed that the proposed algorithms and the method in
\cite{kha06} passed the feasibility test for almost all channel
realizations. Intuitively, group power allocation is more feasible
than the method in \cite{kha06} because the average, instead of
per stream, SINR constraints are easier to be achieved, and they
make Equation (\ref{equ:GroupMinPowerEq}) better conditioned than
the corresponding equation in \cite{kha06}. Thus nonnegative
solutions of (\ref{equ:GroupMinPowerEq}) are easier to be found.
In addition, since the group power allocation is a special case of
the per stream power allocation with equal power distribution
among the streams of a user, the per stream power allocation
method should be even more feasible. As an example, when a high
target SINR ($\gamma=12$ dB) is desired, simulation shows that the
probabilities of feasibility for the group power allocation, the
per stream power allocation, the method in \cite{kha06}, group BD
and per stream BD are 99\%, 100\%, 67\%, 100\%, and 100\%,
respectively. For Case 2, the system can only support lower target
SINR and the probabilities of feasibility for the above five
algorithms when $\gamma=3$ dB are 100\%, 100\%, 0\%, 0\%, and 0\%,
respectively. Note that for Case 2, the BD based methods can not be
applied since $M$ is not large enough.

As for the insights to the convergence behavior, typically
each optimization step improves its objective function as outlined
in Section \ref{sec:duality}. The beamforming step maximizes each
user's sum SINR of the data streams and the power allocation step
optimizes the balanced level. As an example, in Fig.
\ref{fig:IterationVSlevel}, we plot the balanced levels versus
iteration times of the two proposed algorithms for Problem Pr
under the same channel conditions. The total power constraint is
set to 15 dBm, and the SINR constraint $\gamma=1$. The arrows
point at the numbers of iterations where the algorithms meet the
convergence criteria. From the figure, the two proposed algorithms
typically do not oscillate often and exhibit smooth transient
behaviors. We also observe that the convergence behavior of
the group power allocation is slightly better than that of the per
stream power allocation, i.e., the per stream power allocation is
not as smooth as the group power allocation and needs more
iterations to approach the balanced level. The reason why the per
stream power allocation has worse convergence behavior is that
after a power allocation step, the noise whitening property
obtained by the previous maximum SINR filter bank may no longer be
valid, that is, ${{\bf{V}}_{k}^H {\bf{R}}_{n,k}^{{\rm{DL}}}
{\bf{V}}_{k} }$ may no longer be a scaled identity matrix. This
effect may decrease the balanced level. However, in most cases
this negative effect has a small impact on the eventual
performance. Table \ref{tab:IterationTimes} lists the iteration
times needed to converge for both problems. From these results, we
can see that all three methods need more iterations to converge in
Case 2. Note that for Problem Pp, the target SINRs $\gamma$ for
Case 2 are smaller than those of Case 1 since the method in
\cite{kha06} is not feasible for $\gamma \geq 3$ dB. Also, the
method in \cite{kha06} needs significantly more iterations when
$\gamma=2$ dB.

Fig. \ref{fig:ConvProb} shows the probabilities of the
proposed algorithms and the method in \cite{kha06} converging
within 50 iterations given that they have passed the feasibility
test, for the power minimization problem in Case 1 (settings of
Fig. \ref{fig:MinPowerNoCombine}). We can see that the group power
allocation and the method in \cite{kha06} both exhibit good
convergence probabilities while the per stream power allocation
converges better at low SINR than at high SINR. The reason for the
lower convergence probability of the per stream power allocation
is that the linear programming makes the algorithm prone to
oscillation between feasible solutions from iteration to
iteration.
In practice, as long as the solution is a nonnegative power
vector, the SINR constraints are achieved, no matter the algorithm
oscillates or not. Moreover, even when the per stream power
allocation oscillates at the final iterations, typically the SINRs
are still higher than that of the group power allocation. So one
can simply pick the solution at the final iteration and still
obtain a better performance. The other way is to avoid oscillation
by switching to the group power allocation whenever the per stream
power allocation algorithm oscillates. The performance of this
combined algorithm should be between the performance of the per
stream power allocation and the group power allocation. Fig.
\ref{fig:ConvProb_more_user} shows the convergence probability for
Problem Pp in Case 2. Since the method in \cite{kha06} is not
feasible when SINR constraint for $\gamma \geq 3$ dB, we only plot
for $\gamma < 3$ dB. The group power allocation still converges
almost surely in this overloaded case.

\section{Computational Complexity}\label{sec:Discussions}
In Table \ref{tab:Complexity}, we compare the computational
complexity in one iteration based on the number of complex
multiplications. From this table, one can see that there is
no single step which dominates the complexity for each algorithm,
so we list all of them for comparison. For each optimization step,
the complexity of group power allocation (Table
\ref{tab:AlgGroupSumPower}) or power minimization (Table
\ref{tab:AlgGroupMinPower}) is lower than that of the method in
\cite{kha06}, and we have shown in Section
\ref{sec:SimulationResult} that the performances of the proposed
methods are also superior. The reason for
the complexity saving of the
group power allocation method is due to the fact that the method in
\cite{kha06} processes the streams separately (matrix dimension $L$), while the
group power allocation processes the streams of a user jointly (matrix dimension $K$, $K<L$).
For the per stream algorithms (Table \ref{tab:AlgAntennaSumPower}
and \ref{tab:AlgAntennaMinPower}), the complexity of power
allocation is at most slightly higher than that of the method in
\cite{kha06}, but the performance is much better.

In addition to the computational complexity in one iteration, the
average number of iterations needed for convergence also affects
the system complexity. The average numbers of iterations for the
three algorithms in the simulation settings of Fig.
\ref{fig:MinPowerNoCombine} and Fig. \ref{fig:SumPowerNoCombine}
are shown in Table \ref{tab:IterationTimes} (for the power
minimization Problem Pp, the average number of iterations needed
by the feasibility test is included). This table shows that the
group power allocation method has the fastest convergence among
the three algorithms, while the per stream power allocation has
the slowest convergence. Compared to the method in \cite{kha06},
the group power allocation has a lower computational complexity,
converges faster and performs better. If more complicated
computation is allowed, the per stream power allocation exhibits
even better performance.

As for the BD algorithms used in this paper, the computation of the zero-forcing
transmit beamformers has approximately the same complexity as
that of the ``uplink beamforming'' step of Group (Pr) in Table
\ref{tab:Complexity}; while the complexity for receiver
beamformers is approximately the same as that of ``downlink
beamforming'' step of Group (Pr). The complexity of the power
allocation steps is negligible compared with those of the
beamformers. Since the BD algorithms do not need iterations, they are not
listed in the comparisons in Tables \ref{tab:Complexity} and \ref{tab:IterationTimes} (nor in Fig. \ref{fig:ConvProb}).

\section{Conclusion}\label{sec:Conclusion}
Efficient solutions to the joint transmit-receive beamforming and
power allocation under average SINR constraints in the multi-user
MIMO downlink systems were proposed. The beamforming filter is a
GSINR-FB which exploits the intra-group cooperation of grouped
data streams. Due to this selection, the SINR balancing structure
of optimal power allocation holds and simplifies the computation.
Based on the uplink-downlink duality, we formulated the dual
problem in the virtual uplink, and iteratively solved the optimal
beamforming filters and power allocation matrices. The proposed
algorithms are generalizations of the one in \cite{Sch04} to the
scenario with multiple receive antennas per user, and exploit the
receiver diversity more effectively than \cite{kha06}. Simulation
results demonstrated the superiority of the proposed algorithms
over methods based on independent data stream processing
\cite{kha06} and BD in terms of performance. Moreover, the
computational complexities of the proposed methods are comparable with
that of \cite{kha06}.

%

\bibliographystyle{IEEEtran}
\bibliography{IEEEabrv,CodeSN,scpub,icc09ref,WCOMref}
\newpage
\begin{table}
\renewcommand{\arraystretch}{1.1}
\caption{Basic Steps of the $n$th Iteration} \label{tab:AlgBasic}
\centering
\begin{tabular}{ll}
\hline \hline \\
1: & \textit{First Downlink Power Allocation} \\
 & Fixed ${\bf U}^{(n)}$ and ${\bf V}^{(n)}$, find new ${\bf P}^{(2n+1)}$
 \\
2: & \textit{Downlink Receive Maximum SINR Beamforming} \\
 & Fixed ${\bf P}^{(2n+1)}$ and ${\bf U}^{(n)}$, find new ${\bf V}^{(n+1)}$
 \\
3: & \textit{Second Downlink Power Allocation} \\
 & Fixed ${\bf
U}^{(n)}$ and ${\bf V}^{(n+1)}$, find new ${\bf P}^{(2n+2)}$
 \\
4: & \textit{First Virtual Uplink Power Allocation} \\
 & Fixed ${\bf V}^{(n+1)}$ and ${\bf U}^{(n)}$, find new ${\bf Q}^{(2n+1)}$
 \\
5: & \textit{Virtual Uplink Receive Maximum SINR Beamforming} \\
 & Fixed ${\bf Q}^{(2n+1)}$ and ${\bf
V}^{(n+1)}$, find new ${\bf U}^{(n+1)}$
 \\
6: & \textit{Second Virtual Uplink Power Allocation} \\
 & Fixed ${\bf U}^{(n+1)}$ and ${\bf V}^{(n+1)}$, find new ${\bf
Q}^{(2n+2)}$.
 \\ \\
\hline \hline \\
\end{tabular}
\end{table}

\begin{table}
\renewcommand{\arraystretch}{1.1}
\caption{Iterative algorithm for Problem Pr with group power allocation}
\label{tab:AlgGroupSumPower} \centering
\begin{tabular}{ll}
\hline \hline \\
\multicolumn{2}{l}{ \textit{\textbf{Initialization:}}
 $\mathbf{U}=\mathbf{I},\mathbf{V}=\mathbf{I}$} \\
\multicolumn{2}{l}{\textit{\textbf{Iteration:}}} \\
1: & \textit{First Downlink Power Allocation with Sum Power Constraint} \\
 & Solve $\bf p$ in
$\mathbf{\Upsilon}
    \begin{pmatrix}
        \mathbf{p} \\
        1
    \end{pmatrix}
 = \frac{1}{{C^{{\rm{DL}}} }}\begin{pmatrix}
        \mathbf{p} \\
        1
    \end{pmatrix}$
 \\
2: & \textit{Downlink Receive Maximum SINR Beamforming} \\
 &
 for $k=1:K$ \\
 &
 \hspace{5 mm}$\mathbf{V}_k=eig(\mathbf{R}_{s,k}^\mathrm{DL},\mathbf{R}_{n,k}^\mathrm{DL})$
 \\
3: & \textit{Second Downlink Power Allocation with Sum Power Constraint} \\
 & Solve $\bf p$ in
$\mathbf{\Upsilon} \begin{pmatrix}
        \mathbf{p} \\
        1
    \end{pmatrix} = \frac{1}{{C^{{\rm{DL}}} }}\begin{pmatrix}
        \mathbf{p} \\
        1
    \end{pmatrix}$
 \\
4: & \textit{First Virtual Uplink Power Allocation with Sum Power Constraint} \\
 & Solve $\bf q$ in
$\mathbf{\Lambda } \begin{pmatrix}
        \mathbf{q} \\
        1
    \end{pmatrix} = \frac{1}{{C^{{\rm{UL}}} }}\begin{pmatrix}
        \mathbf{q} \\
        1
    \end{pmatrix}$
 \\
5: & \textit{Virtual Uplink Receive Maximum SINR Beamforming} \\
 &
 for $k=1:K$ \\
 &
 \hspace{5 mm}$\mathbf{U}_k=eig(\mathbf{R}_{s,k}^\mathrm{UL},\mathbf{R}_{n,k}^\mathrm{UL})$
 \\
6: & \textit{Second Virtual Uplink Power Allocation with Sum Power Constraint} \\
 & Solve $\bf q$ in
$\mathbf{\Lambda } \begin{pmatrix}
        \mathbf{q} \\
        1
    \end{pmatrix} = \frac{1}{{C^{{\rm{UL}}} }}\begin{pmatrix}
        \mathbf{q} \\
        1
    \end{pmatrix}$
 \\
7: & \textit{Repeat steps 1-6 until convergence, i.e.,
$|C^{{\rm DL}(n)}-C^{{\rm DL}(n-1)}|< \epsilon $} \\
\hline \hline \\
\end{tabular}
\end{table}

\begin{table}
\renewcommand{\arraystretch}{1.1}
\caption{Iterative algorithm for Problem Pp with group power allocation} \label{tab:AlgGroupMinPower}
\centering
\begin{tabular}{ll}
\hline \hline \\
\multicolumn{2}{l}{ \textit{\textbf{Initialization:}} Feasibility
test using the algorithm in Table
\ref{tab:AlgGroupSumPower}}, if failure then exit. \\

\multicolumn{2}{l}{ \textit{\textbf{Iteration:}}} \\
1: & \textit{First Downlink Power Minimization} \\

 & ${\bf{p}} = ({\bf{I}} - {\bf{D\Psi }})^{-1} {\bf{D\bm \sigma }}$\\

2: & \textit{Downlink Receive Maximum SINR Beamforming} \\
 &
 for $k=1:K$ \\
 &
 \hspace{5 mm}$\mathbf{V}_k=eig(\mathbf{R}_{s,k}^\mathrm{DL},\mathbf{R}_{n,k}^\mathrm{DL})$
 \\
3: & \textit{Second Downlink Power Minimization} \\
 & ${\bf{p}} = ({\bf{I}} - {\bf{D\Psi }})^{-1} {\bf{D\bm \sigma }}$\\

4: & \textit{First Virtual Uplink Power Minimization} \\
 &  ${\bf{q}} = ({\bf{I}} - {\bf{D\Psi }}^T )^{ - 1} {\bf{D\bm \sigma}}$\\

5: & \textit{Virtual Uplink Receive Maximum SINR Beamforming} \\
 &
 for $k=1:K$ \\
 &
 \hspace{5 mm}$\mathbf{U}_k=eig(\mathbf{R}_{s,k}^\mathrm{UL},\mathbf{R}_{n,k}^\mathrm{UL})$
 \\
6: & \textit{Second Virtual Uplink Power Minimization} \\
 &  ${\bf{q}} = ({\bf{I}} - {\bf{D\Psi }}^T )^{ - 1} {\bf{D\bm \sigma}}$\\

7: & \textit{Repeat steps 1-6 until convergence, i.e.,
$|C^{{\rm DL}}-1|< \epsilon $} \\
\hline \hline \\
\end{tabular}
\end{table}

\begin{table}
\renewcommand{\arraystretch}{1}
\caption{Iterative algorithm for Problem Pp with per stream power allocation}
\label{tab:AlgAntennaMinPower} \centering
\begin{tabular}{ll}
\hline \hline \\
\multicolumn{2}{l}{ \textit{\textbf{Initialization:}}
Feasibility test using the algorithm in Table \ref{tab:AlgAntennaSumPower}}, if failure then exit.\\
\multicolumn{2}{l}{\textit{\textbf{Iteration:}}} \\
1: & \textit{First Downlink Power Minimization} \\
 & {Solve $\bf p$ in the linear programming problem:} \\
 &
$\begin{array}{l}
 \min {\rm{\ }}{\bf{1}}^T {\bf{p}} \\
 {\rm{s.t.\ \ }}\sum\limits_{l = 1}^{L_k } {p_{kl} [{\bf{A}}_{kk} ]_{ll}/\gamma_k }  - \sum\limits_{j=1,j\ne k}^K {\sum\limits_{l = 1}^{L_j } {p_{jl} [{\bf{A}}_{jk} ]_{ll} } }  = L_k \sigma^2,\\
 \hspace{7mm}{\rm{for\ }}k = 1,...,K, {\rm{\ and\ }} \mathbf{p} \ge_e 0\\
 \end{array}$
 \\
2: & \textit{Downlink Receive Maximum SINR Beamforming} \\
 &
 for $k=1:K$ \\
 &
 \hspace{5 mm}$\mathbf{V}_k=eig(\mathbf{R}_{s,k}^\mathrm{DL},\mathbf{R}_{n,k}^\mathrm{DL})$
 \\
3: & \textit{Second Downlink Power Minimization} \\
 & {Solve $\bf p$ in the linear programming problem:} \\
 &
$\begin{array}{l}
 \min {\rm{\ }}{\bf{1}}^T {\bf{p}} \\
 {\rm{s.t.\ \ }}\sum\limits_{l = 1}^{L_k } {p_{kl} [{\bf{A}}_{kk} ]_{ll}/\gamma_k }  - \sum\limits_{j=1,j\ne k}^K {\sum\limits_{l = 1}^{L_j } {p_{jl} [{\bf{a}}_{jk} ]_{ll} } }  = L_k \sigma^2,\\
 \hspace{7mm}{\rm{for\ }}k = 1,...,K, {\rm{\ and\ }} \mathbf{p} \ge_e 0\\
 \end{array}$
 \\
4: & \textit{First Virtual Uplink Power Minimization} \\
 & {Solve $\bf q$ in the linear programming problem:} \\
 &
 $\begin{array}{l}
 \min {\rm{\ }}{\bf{1}}^T \mathbf{q} \\
 {\rm{s.t.\ \ }}\sum\limits_{l = 1}^{L_k } {q_{kl} [{\bf{B}}_{kk} ]_{ll}/\gamma_k }  - \sum\limits_{j=1,j\ne k}^K {\sum\limits_{l = 1}^{L_j } {q_{jl} [{\bf{B}}_{kj} ]_{ll} } }  = L_k \sigma^2,\\
 \hspace{7mm}{\rm{for\ }}k = 1,...,K, {\rm{\ and\ }}\mathbf{q} \ge_e 0 \\
 \end{array}$
 \\
5: & \textit{Virtual Uplink Receive Maximum SINR Beamforming} \\
 &
 for $k=1:K$ \\
 &
 \hspace{5 mm}$\mathbf{U}_k=eig(\mathbf{R}_{s,k}^\mathrm{UL},\mathbf{R}_{n,k}^\mathrm{UL})$
 \\
6: & \textit{Second Virtual Uplink Power Minimization} \\
 & {Solve $\bf q$ in the linear programming problem:} \\
 &
 $\begin{array}{l}
 \min {\rm{\ }}{\bf{1}}^T \mathbf{q} \\
 {\rm{s.t.\ \ }}\sum\limits_{l = 1}^{L_k } {q_{kl} [{\bf{B}}_{kk} ]_{ll}/\gamma_k }  - \sum\limits_{j=1,j\ne k}^K {\sum\limits_{l = 1}^{L_j } {q_{jl} [{\bf{B}}_{kj} ]_{ll} } }  = L_k \sigma^2,\\
 \hspace{7mm}{\rm{for\ }}k = 1,...,K, {\rm{\ and\ }}\mathbf{q} \ge_e 0 \\
 \end{array}$
 \\
7: & \textit{Repeat steps 1-6 until convergence, i.e., $|C^{{\rm DL}}-1|< \epsilon $} \\
\hline \hline \\
\end{tabular}
\end{table}

\begin{table}
\renewcommand{\arraystretch}{1}
\caption{Iterative algorithm for Problem Pr with per stream power allocation}
\label{tab:AlgAntennaSumPower} \centering
\begin{tabular}{ll}
\hline \hline \\
\multicolumn{2}{l}{ \textit{\textbf{Initialization:}}
 $\mathbf{U}=\mathbf{I},\mathbf{V}=\mathbf{I}$} \\
\multicolumn{2}{l}{\textit{\textbf{Iteration:}}} \\
1: & \textit{First Downlink Power Allocation with Sum Power Constraint} \\
 &
 $t_k^{(2n + 1)}  = \frac{{P_{\max } }}{{G_k^{(2n)} \sum\limits_{j =
1}^K {\frac{1}{{G_j^{(2n)} }}} }}$, where $G_k^{(2n)}  =
\frac{{g_{kk}/\gamma_k }}{{\sum\limits_{j \ne k} {t_j^{(2n)} g_{jk} }
+ L_k \sigma ^2 }}$.
 \\
2: & \textit{Downlink Receive Maximum SINR Beamforming} \\
 &
 for $k=1:K$ \\
 &
 \hspace{5 mm}$\mathbf{V}_k=eig(\mathbf{R}_{s,k}^\mathrm{DL},\mathbf{R}_{n,k}^\mathrm{DL})$
 \\
3: & \textit{Second Downlink Power Allocation with Sum Power Constraint} \\
 &
 $t_k^{(2n + 2)}  = \frac{{P_{\max } }}{{G_k^{(2n+1)} \sum\limits_{j =
1}^K {\frac{1}{{G_j^{(2n+1)} }}} }}$, where $G_k^{(2n+1)}  =
\frac{{g_{kk}/\gamma_k }}{{\sum\limits_{j \ne k} {t_j^{(2n+1)} g_{jk} }
+ L_k \sigma ^2 }}$.
 \\
4: & \textit{First Virtual Uplink Power Allocation with Sum Power Constraint} \\
 &
 $s_k^{(2n + 1)}  = \frac{{P_{\max } }}{{H_k^{(2n)} \sum\limits_{j =
1}^K {\frac{1}{{H_j^{(2n)} }}} }}$, where $H_k^{(2n)}  =
\frac{{h_{kk}/\gamma_k }}{{\sum\limits_{j \ne k} {s_j^{(2n)}
h_{jk} } + L_k \sigma ^2 }}$.
 \\
5: & \textit{Virtual Uplink Receive Maximum SINR Beamforming} \\
 &
 for $k=1:K$ \\
 &
 \hspace{5 mm}$\mathbf{U}_k=eig(\mathbf{R}_{s,k}^\mathrm{UL},\mathbf{R}_{n,k}^\mathrm{UL})$
 \\
6: & \textit{Second Virtual Uplink Power Allocation with Sum Power Constraint} \\
 &
 $s_k^{(2n + 2)}  = \frac{{P_{\max } }}{{H_k^{(2n+1)} \sum\limits_{j =
1}^K {\frac{1}{{H_j^{(2n+1)} }}} }}$, where $H_k^{(2n+1)}  =
\frac{{h_{kk}/\gamma_k }}{{\sum\limits_{j \ne k} {s_j^{(2n+1)}
h_{jk} }  + L_k \sigma ^2 }}$.
 \\
7: & \textit{Repeat steps 1-6 until convergence, i.e., $|C^{{\rm DL}(n)}-C^{{\rm DL}(n-1)}|< \epsilon $                 } \\
\hline \hline \\
\end{tabular}
\end{table}


\begin{table}
    \renewcommand{\arraystretch}{1.2}
    \caption{Complexities of the optimization steps in one iteration}
    \label{tab:Complexity} \centering
    \begin{tabular}{@{}lp{2.2cm}p{2.2cm}p{2.2cm}p{2.2cm}@{}}  
    \toprule
     & \sf Uplink Beamforming & \sf Uplink Power Allocation & \sf Downlink Beamforming &  \sf Downlink Power Allocation\\

    \midrule
    \sf Group (Pr) & $O(KM^3)$ & $O((K+1)^3)$ & $O(\sum_{k=1}^{K}L_k^3)$ & $O((K+1)^3)$ \\

    \sf Group (Pp) & $O(KM^3)$ & $O(K^3)$ & $O(\sum_{k=1}^{K}L_k^3)$ & $O(K^3)$ \\

    \sf Per Stream (Pr) & $O(KM^3)$ & $O(\sum_{k=1}^{K}L_k^2 M)$ & $O(\sum_{k=1}^{K}L_k^3)$ & $O(\sum_{k=1}^{K}L_k^2 M)$\\

    \sf Per Stream (Pp) & $O(KM^3)$ & $O(L^{3.5})$ & $O(\sum_{k=1}^{K}L_k^3)$ & $O(L^{3.5})$ \\

    \sf Khachan's  (Pr) & $O(LM^3)$ & $O(L^3)$ & $O(\sum_{k=1}^{K}L_k^4)$ & $O(L^3)$ \\

    \sf Khachan's  (Pp) & $O(LM^3)$ & $O(L^3)$ & $O(\sum_{k=1}^{K}L_k^4)$ & $O(L^3)$ \\
    \bottomrule
    \end{tabular}
\end{table}

\begin{table}
    \renewcommand{\arraystretch}{1.2}
    \caption{Average numbers of iterations needed for convergence}
    \label{tab:IterationTimes} \centering

    \begin{tabular}{@{}lp{1cm}p{1cm}p{1cm}llp{1cm}p{1cm}p{1cm}@{}}  
    \multicolumn{9}{c}{Case 1 (Fully loaded)}\\
    \multicolumn{4}{c}{Problem Pp}& &\multicolumn{4}{c}{Problem Pr}\\
    \cline{1-4} \cline{6-9}
     SINR constraint ($\gamma$) & 2 &  4 & 6 & & Pmax(dB) & 10 & 12 & 14\\
    \cline{1-4} \cline{6-9}
    \sf Group &   9.6500  &   10.7200 &   11.7200& & &12.359 & 12.608 & 12.558\\
    \sf Khachan's &     15.3720  &  17.7600  & 21.7200& & & 13.857 & 13.316 & 13.382\\
    \sf Per Stream &     25.7861   &  26.1097  &  30.2598& & & 15.475 & 14.871 & 13.906\\
    \cline{1-4} \cline{6-9}
    \\
    \multicolumn{9}{c}{Case 2 (Over loaded)}\\
    \multicolumn{4}{c}{Problem Pp}& &\multicolumn{4}{c}{Problem Pr}\\
    \cline{1-4} \cline{6-9}
     SINR constraint ($\gamma$) & -2 &  0 & 2 & & Pmax(dB) & 10 & 12 & 14\\
    \cline{1-4} \cline{6-9}
    \sf Group &   10.41  &   11.88 &   14.96& & &15.43 & 17.37 & 20.31\\
    \sf Khachan's &     16.63  &  22.97  & 40.93& & & 14.11 & 14.58 & 15.72\\
    \sf Per Stream &     25.78   &  26.10  &  30.25& & & 17.53 & 18.3 & 17.04\\
    \cline{1-4} \cline{6-9}
    \\
    \end{tabular}

\end{table}

\begin{figure}
  \centering
  \includegraphics[scale=0.6]{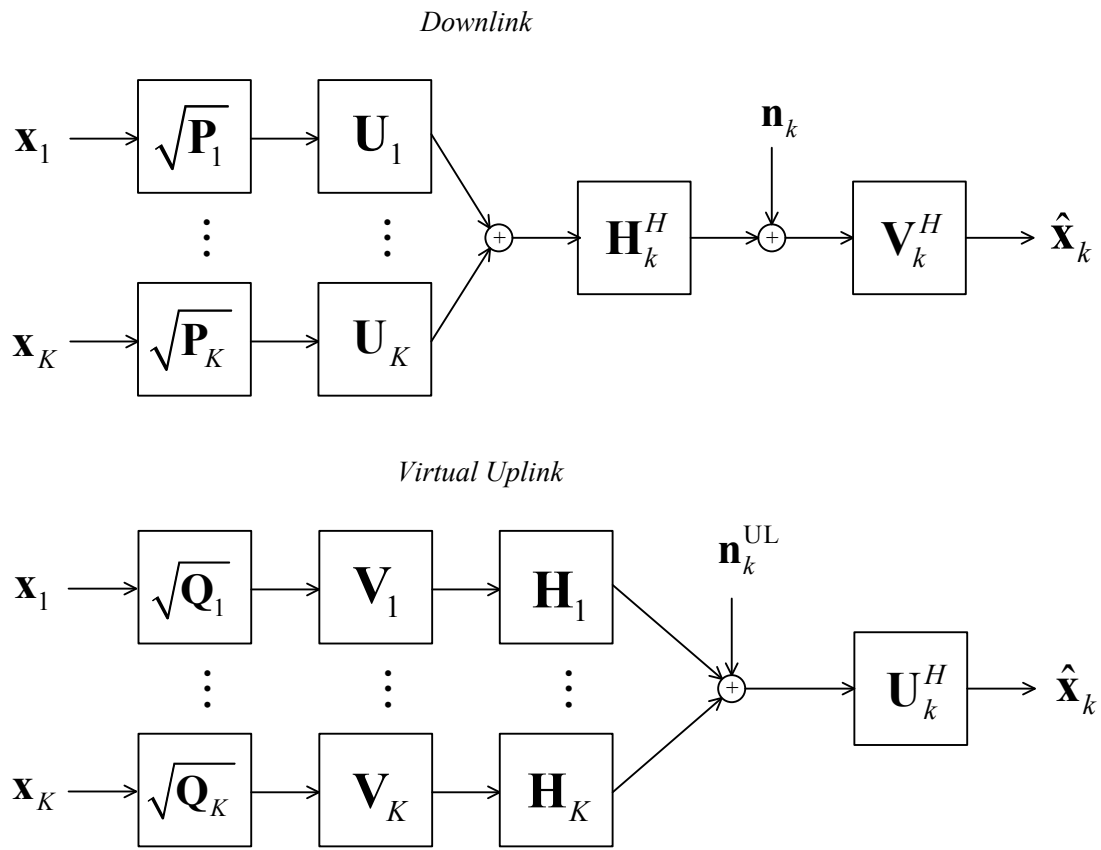}\\
  \caption{ MIMO downlink system model for user $k$ and its virtual uplink.}\label{fig:DL_UL}
\end{figure}

\begin{figure}
  \centering
  \includegraphics[scale=0.9]{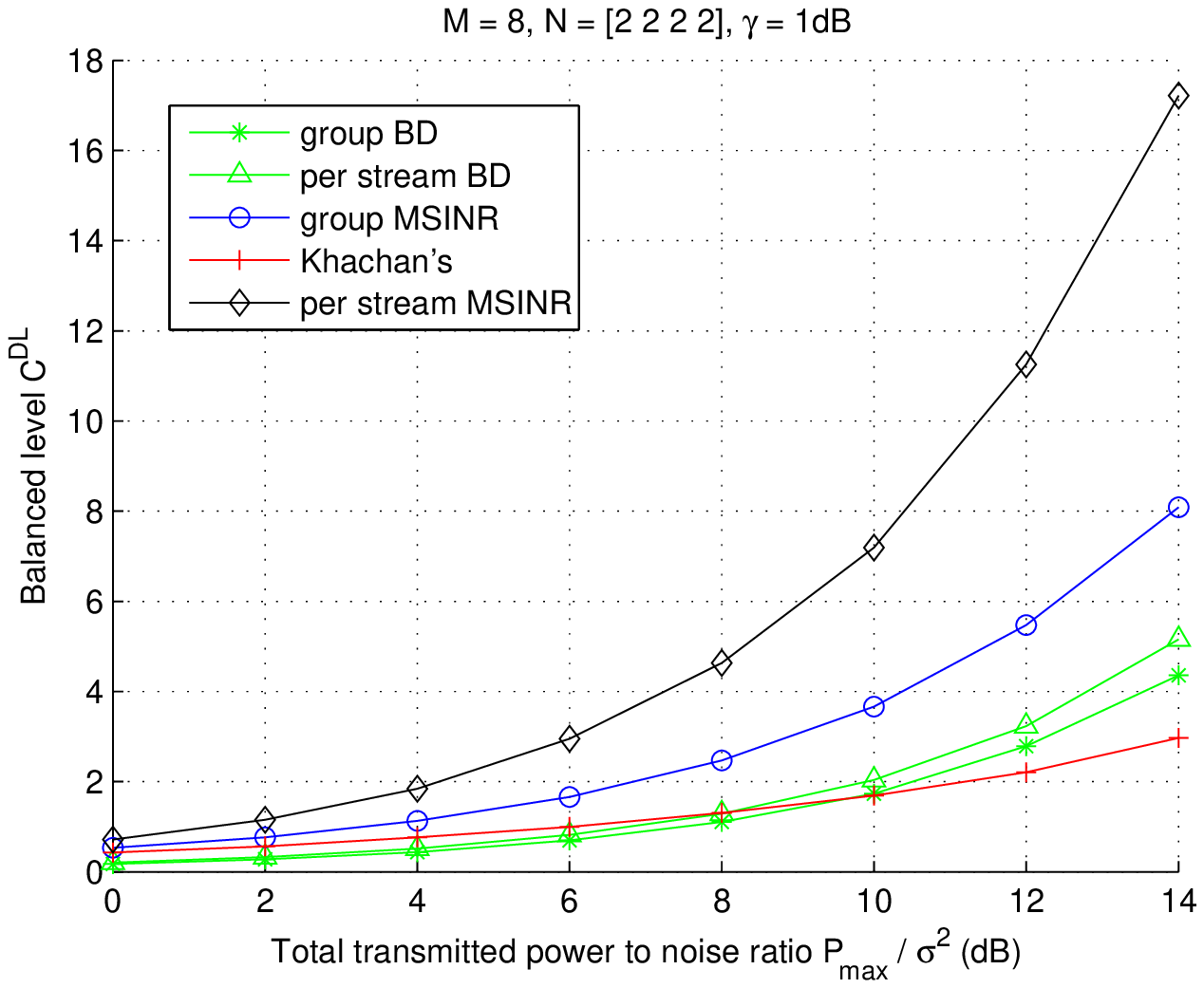}\\
  \caption{Comparison with \cite{kha06} and BD for Problem Pr. $K=4$, $M=8$, $N_k=2, \;\forall k$.}\label{fig:SumPowerNoCombine}
\end{figure}

\begin{figure}
  \centering
  \includegraphics[scale=0.9]{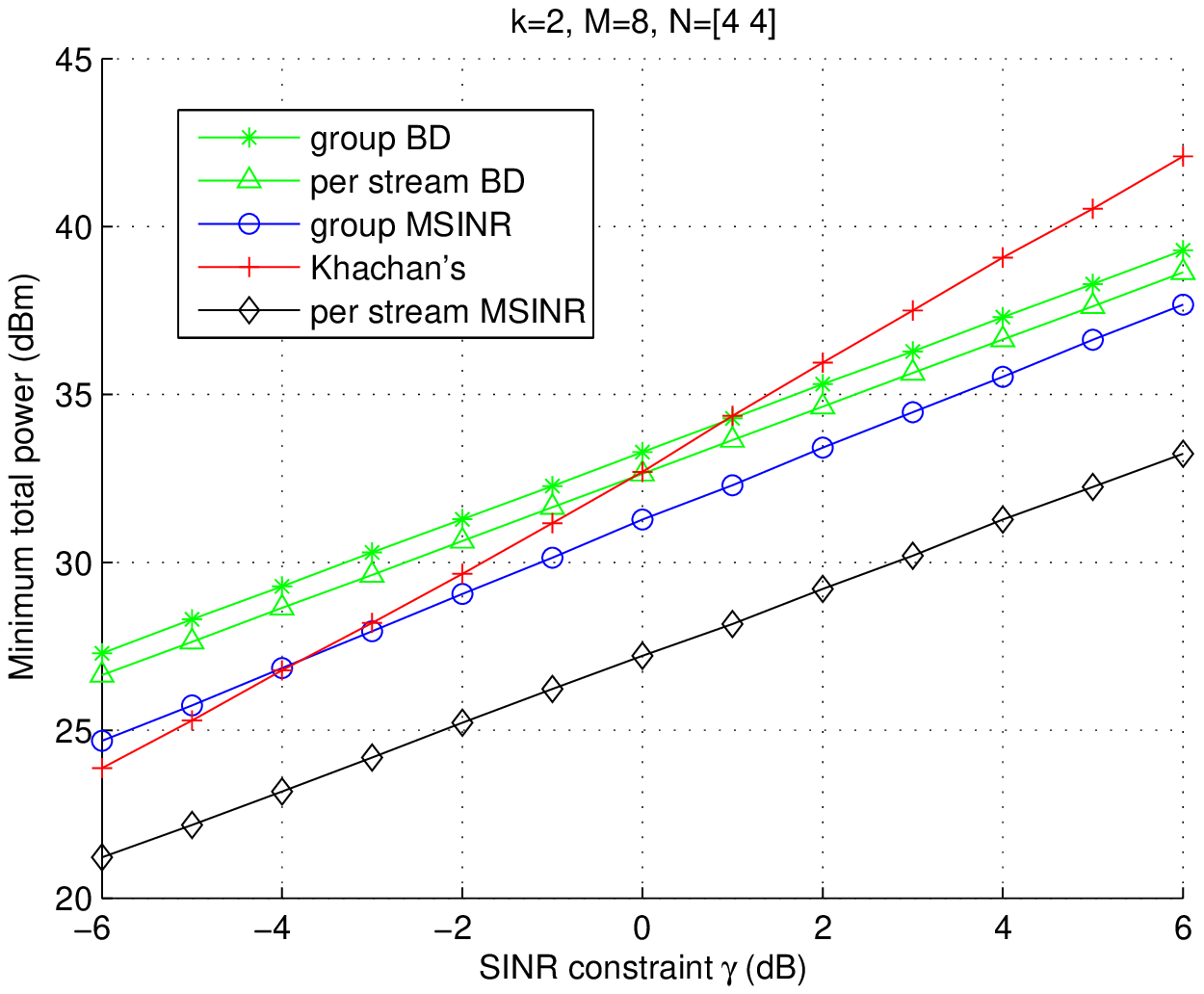}\\
  \caption{Comparison with \cite{kha06} and BD for Problem Pp. $K=2$, $M=8$, $N_k=4, \;\forall k$.}\label{fig:MinPowerNoCombine}
\end{figure}

\begin{figure}
  \centering
  \includegraphics[scale=0.9]{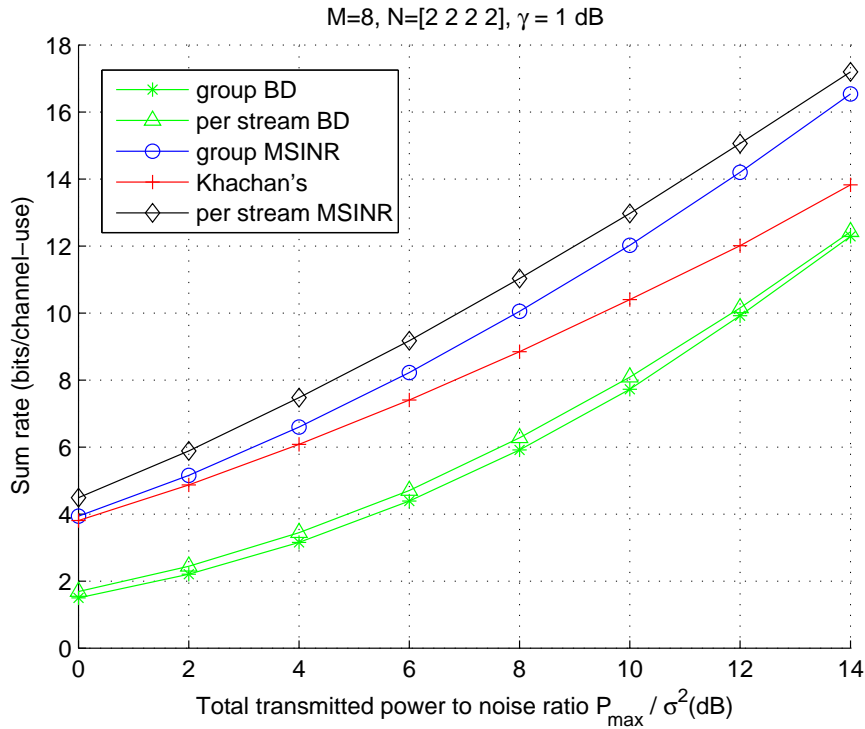}\\
  \caption{Sum rate comparison with \cite{kha06} and BD for Problem Pr. $K=4$, $M=8$, $N_k=2, \;\forall k$.}\label{fig:SumRate}
\end{figure}

\begin{figure}
  \centering
  \includegraphics[scale=0.85]{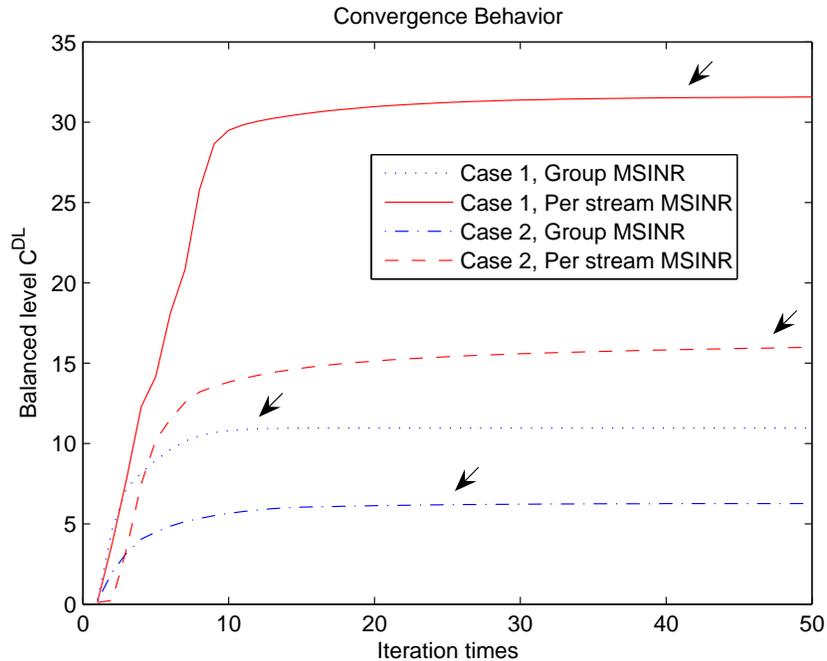}\\
  \caption{Convergence behaviors of the proposed algorithms for Problem Pr in Case 1 (fully loaded) and Case 2 (over loaded). The
arrows point at the numbers of iterations where the algorithms meet the
convergence criteria.}\label{fig:IterationVSlevel}
\end{figure}

\begin{figure}
  \centering
  \includegraphics[scale=0.9]{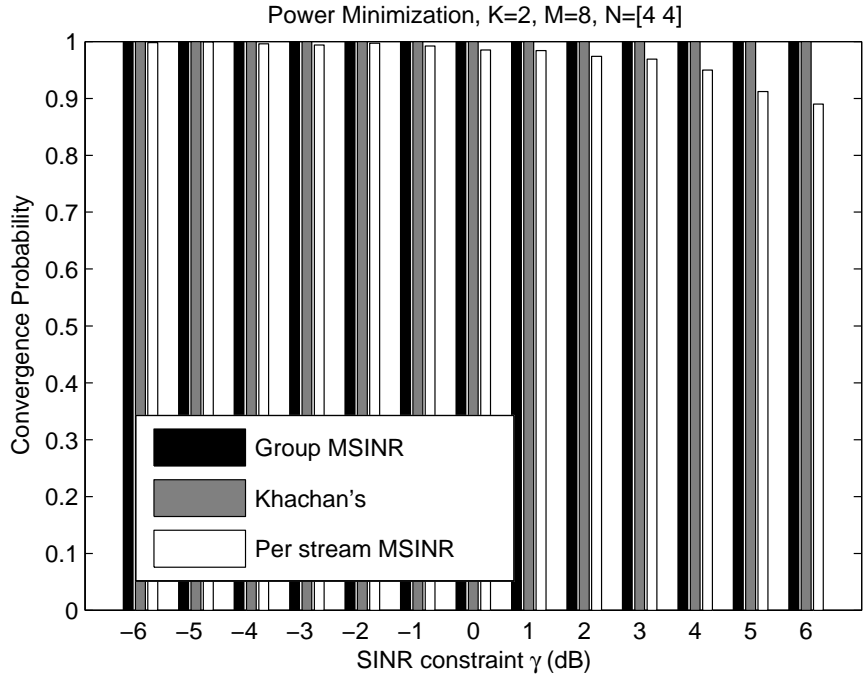}\\
  \caption{Comparison of the convergence probabilities for Problem Pp in Case 1 (fully loaded): $K=2$, $M=8$, $N_k=4, \;\forall k$.}
  \label{fig:ConvProb}
\end{figure}

\begin{figure}
  \centering
  \includegraphics[scale=0.9]{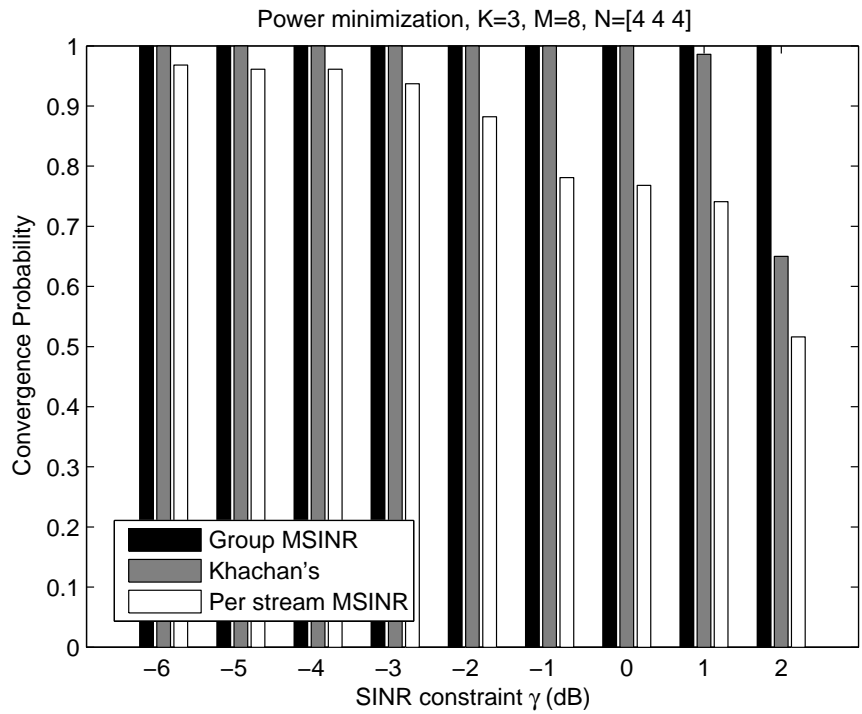}\\
  \caption{Comparison of the convergence probabilities for Problem Pp in Case 2 (over loaded): $K=3$, $M=8$, $N_k=4, \;\forall k$.}
  \label{fig:ConvProb_more_user}
\end{figure}

\end{document}